%% file: main.tex
\documentclass[conference]{IEEEtran}
\usepackage[utf8]{inputenc}

\makeatletter
\newcommand{\linebreakand}{%
  \end{@IEEEauthorhalign}
  \hfill\mbox{}\par
  \mbox{}\hfill\begin{@IEEEauthorhalign}
}
\makeatother

\input{preamble}

\begin{document}

\title{ShorTor: Improving Tor Network Latency via Multi-hop Overlay Routing}

\author{
\IEEEauthorblockN{Kyle Hogan}
\IEEEauthorblockA{\textit{MIT CSAIL}}
\and
\IEEEauthorblockN{Sacha Servan-Schreiber}
\IEEEauthorblockA{\textit{MIT CSAIL}}
\and
\IEEEauthorblockN{Zachary Newman}
\IEEEauthorblockA{\textit{MIT CSAIL}}
\and
\IEEEauthorblockN{Ben Weintraub}
\IEEEauthorblockA{\textit{Northeastern University}}
\linebreakand
\IEEEauthorblockN{Cristina Nita-Rotaru}
\IEEEauthorblockA{\textit{Northeastern University}}
\and
\IEEEauthorblockN{Srinivas Devadas}
\IEEEauthorblockA{\textit{MIT CSAIL}}
}

\maketitle

\input{sections/abstract}


\section{Introduction}
\label{sec:intro}
\input{sections/intro}

\section{Background}
\label{sec:background}
\input{sections/background}

\section{ShorTor}
\label{sec:shortor}
\input{sections/shortor}

\section{Evaluation}
\label{sec:measurements}
\input{sections/measurements}

\section{Security Analysis}
\label{sec:secanalysis}
\input{sections/security}

\section{Related Work}
\label{sec:relatedwork}
\input{sections/relatedwork}

\section{Discussion}
\label{sec:discussion}
\input{sections/discussion}

\section{Conclusion}
\label{sec:conclusion}
\input{sections/conclusion}

\section{Acknowledgements}
\label{sec:ack}
We would like to thank Arthur Berger for invaluable suggestions and discussion surrounding internet routing, Frank Cangialosi for conversations about Ting and tips on measuring the Tor network, and Anish Athalye for his ideas around evaluating page load times and feedback on our security analysis.
We are also grateful to the Tor community for their support of this project and, in particular, Roger Dingledine, Georg Koppen, and Tariq Elahi for extensive advice related to running our experiments. 

\bibliographystyle{plainnat}
\bibliography{main}

\appendices
\crefalias{section}{appendix}  

\newpage
\balance
\input{sections/appendix}

\end{document}

%% file: sections/abstract.tex
\begin{abstract}
We present \name, a protocol for reducing latency on the Tor network.
\name uses \emph{multi-hop overlay routing}, a technique typically employed by content delivery networks, to influence the route Tor traffic takes across the internet.
In this way, \name avoids slow paths and improves the experience for end users by reducing the latency of their connections while imposing minimal bandwidth overhead.

\name functions as an overlay on top of onion routing---Tor's existing routing protocol---and is run by Tor relays, making it independent of the path selection performed by Tor clients.
As such, \name reduces latency while preserving Tor's existing security properties.
Specifically, the routes taken in \name are in no way correlated to either the Tor user or their destination, including the geographic location of either party.
We analyze the security of \name using the AnoA framework, showing that \name maintains all of Tor's anonymity guarantees. 
We augment our theoretical claims with an empirical analysis.

To evaluate \name's performance, we collect a real-world dataset of over \num{400000} latency measurements between the \num{1000} most popular Tor relays, which collectively see the vast majority of Tor traffic.
With this data, we identify pairs of relays that could benefit from \name: that is, two relays where introducing an additional intermediate network hop results in \emph{lower} latency than the direct route between them.
We use our measurement dataset to simulate the impact on end users by applying \name to two million Tor circuits chosen according to Tor's specification.

\name reduces the latency for the 99\textsuperscript{th} percentile of relay pairs in Tor by \SI{148}{\milli\second}.
Similarly, \name reduces the latency of Tor circuits by \SI{122}{\milli\second} at the 99\textsuperscript{th} percentile.
In practice, this translates to \name truncating tail latencies for Tor which has a direct impact on page load times and, consequently, user experience on the Tor browser. 
\end{abstract}

%% file: sections/intro.tex
Tor is the foremost deployed system for anonymous communication.
Millions of people around the world use Tor every day to escape censorship and avoid surveillance of their browsing habits~\cite{tor-paper,tor-metrics}.
This broad user base is a critical component of Tor's privacy guarantees.
Tor users are anonymous only amongst each other---not within the general internet population. 
That is, an internet censor may be able to know that \emph{some} Tor user visited a blocked site, but not \emph{which} Tor user.
Because of this, the degree of anonymity Tor provides in practice grows with the total number of concurrent users on the network~\cite{dingledine2006anonymity}. 

This relationship between the privacy of individual users and the overall popularity of Tor makes user experience a major concern for Tor.
A poor experience relative to non-private browsing results in lower adoption of Tor and, ultimately, limits the degree of anonymity Tor is capable of providing. 
A major factor contributing to positive user experience is latency. 
Internet users are very sensitive to latency, and increased page load times discourage user interaction~\cite{attig2017system,arapakis2014impact,crescenzi2016impacts}.
Unfortunately, anonymous communication incurs higher latency than typical internet connections~\cite{van2015vuvuzela,piotrowska2017loopix,angel2016unobservable,tyagi2017stadium,kwon2017atom,lazar2018karaoke,kwon2020xrd,eskandarian2019express}.

In Tor, much of this overhead is due to the underlying structure of its connections~\cite{tor-paper,tor-spec}. 
Tor is a network composed of ${\sim}\num{7000}$ volunteer-run servers, or \emph{relays}, used to route client traffic.
Rather than connecting directly to their destination, Tor clients tunnel their traffic through a series of Tor relays in a process known as \emph{onion routing}.
This drastically increases the path length for Tor traffic, and, in turn, latency.

A substantial body of prior work aims to reduce latency in Tor by changing the relay selection process~\cite{alsabah2013path,imani2019modified,rochet2020claps,yang2015mtor,akhoondi2012lastor,wang2012congestion,sherr2009scalable,barton2018towards,annessi2016navigator}.
By default, Tor clients select relays for their circuit at random, weighted by relay bandwidth, and do not consider path length or circuit latency in the process. 
In contrast, proposals that aim to reduce latency often prioritize selecting circuits that have low latency between relays~\cite{imani2019modified,rochet2020claps,akhoondi2012lastor,sherr2009scalable,barton2018towards,annessi2016navigator}.
Unfortunately, preferentially choosing circuits in this way \emph{also} selects relays that are correlated with the identity of the user or their destination~\cite{mator,anoa-mators-thesis,anoa}.
Many attacks show how this can be exploited to deanonymize Tor users, allowing a passive observer to identify information about user locations~\cite{rochet2020claps,mitseva2020security,mator,anoa-mators-thesis,anoa,geddes2013how,wacek2013empirical,Wan2019guard}.

In this paper, we propose \name, an entirely different approach to reducing the latency of Tor traffic. 
Rather than alter the circuit selection process, \name exploits a technique used by content delivery networks (CDNs) known as \emph{multi-hop overlay routing}~\cite{su2009drafting,bornstein2013optimal}.
Multi-hop overlay routing, like Tor, functions by introducing intermediate hops into its connections, but does so for the explicit purpose of \emph{reducing latency}.
In the wild, CDNs use multi-hop overlay routing to influence the path internet traffic takes.
They do this by inserting their own servers as intermediate points in client connections, avoiding slow default routes by forcing traffic to travel through their server, rather than directly to its destination.
The success of this technique is due to the existence of sub-optimal default routes across the internet~\cite{duffield2007measurement} and the distributed nature of CDN-controlled nodes. 
The broad presence of CDN controlled servers gives them \emph{many} possible routes to choose from, and consequently, increases their odds of finding a faster route to send traffic through. 
In practice, this allows CDNs to avoid outages, congestion, or other delays along the default path.

\smallskip
\noindent
With \name, we ask: 
\begin{displayquote}
\emph{Can multi-hop overlay routing reduce latency in Tor without compromising anonymity?}
\end{displayquote}

While multi-hop overlay routing is widely successful for CDNs (which share some similarities to the Tor network), Tor is much smaller with ${\sim}\num{7000}$ nodes~\cite{tor-metrics-network-size} compared to the ${\sim}\num{300000}$ operated by a CDN~\cite{akamai-cdn}.
In addition to the difference in scale, Tor relays are volunteer run and their placement is not optimized for fast routing.
\name is the first proposal to apply and analyze the impact of multi-hop overlay routing on the Tor network, and is likely of independent interest to other distributed communication systems. 

\subsection{\name}
To reduce the latency experienced by end users of Tor, \name uses multi-hop routing as an additional overlay layer on top of Tor's onion routing protocol.
Crucially, \name is independent of Tor's circuit selection algorithm and the client, operating only between relays.
\name introduces additional hops, which we call \emph{via} relays, that tunnel traffic between relays on a Tor circuit.
Acting as a via is simply a \textit{role} that a normal Tor relay may take in addition to its usual function on circuits.
Via relays, unlike circuit relays, can be introduced after circuit establishment in response to changing network conditions without client involvement or any modification to the circuit itself.
While the basic idea of \name is simple in retrospect, multi-hop overlay routing has security implications for anonymous communication that are not present in CDNs. 

\paragraph{Security}
We demonstrate that \name can find faster paths across Tor \emph{without} the loss in anonymity experienced by other approaches. 
\name selects via relays based \emph{solely} on the adjacent circuit relays.
This process ensures that malicious vias cannot lie about their performance to artificially increase their selection probability.
Specifically, \name operates as an overlay routing layer, requiring no modification to Tor's onion routing or encryption, preserving Tor's security guarantees.
We provide a formal security analysis of the impact \name has on Tor's anonymity using the AnoA framework, which was introduced by \citet{anoa} to analyze the anonymity guarantees of Tor~\cite{mator,backes2014nothing}.
Using AnoA, we show that \name has minimal impact on security when compared to baseline Tor.
However, we find that when used in conjunction with alternative,\footnote{Tor's only deployed path selection algorithm is independent of user location.} location-aware path selection algorithms such as LASTor \cite{akhoondi2012lastor}, \name can exacerbate the existing leakage.
We validate these claims through an empirical analysis on data collected from the Tor network.

\paragraph{Latency Measurements}
To quantify the benefits of \name, we conduct latency measurements between approximately \num{400000} pairs of the \num{1000} most popular Tor relays.
We collect measurements ourselves, rather than use a general-purpose source for internet measurements such as RIPE Atlas~\cite{ripeatlas}, for two main reasons.
(1) internet routing operates at a scale and complexity that cannot easily be simulated~\cite{schlinker2019peering} and
(2) ISPs often treat Tor packets differently from other internet traffic~\cite{cangialosi2015ting}.
Using our own pairwise latency dataset we determine that, despite being much smaller than a typical CDN, Tor can still benefit from multi-hop overlay routing.

\paragraph{Ethics}
Our measurements were conducted on the live Tor network, but did \emph{not} involve any observations on Tor users or their traffic.
We underwent Tor's security review process and followed best practices to limit our impact on the Tor network.
Details can be found in Section \ref{sec:ethics}. 

\paragraph{Practicality}
While \name does require modifications to Tor relays, it does \emph{not} rely on participation of all, or even a majority of, relays and makes no assumptions about or modifications to client behavior.
Tor circuits can benefit from multi-hop overlay routing as long as any two adjacent relays on the path both support it.
The majority of our evaluation assumes that only the \num{1000} most popular Tor relays participate, but we find \name is beneficial with even fewer relays participating.
\name achieves a latency reduction of \SI{178}{\milli\second} at the 99.9\textsuperscript{th} percentile with only the \num{500} most popular relays supporting the protocol. 
As such, \name can be deployed incrementally and still provide meaningful reductions of tail latency on Tor.

\paragraph{Limitations}
Our dataset of pairwise latencies was collected from the \num{1000} most popular Tor relays. 
While these relays do see the majority of traffic on Tor~\cite{greubel2020quantifying}, they are not representative of the full network. 
The less popular relays, while not as likely to be included in circuits, may benefit similarly from \name, and could broaden the pool of available via relays.
A deployed version of \name, however, would naturally include all available relays regardless of popularity. 
As such, the scale of our dataset is strictly a limitation of our evaluation, not of \name's effectiveness in practice.

Using this data, we find that \name primarily impacts \emph{tail} latencies on the Tor network.
On average, \name reduced the RTT between a pair of relays from \SI{42.6}{\milli\second} to \SI{23.5}{\milli\second}, while at the 99.9\textsuperscript{th} percentile the RTT dropped much more substantially from \SI{487}{\milli\second} to \SI{125}{\milli\second}. 
As a result, the speedups \name offers disproportionally benefit a relatively small fraction of Tor users---approximately \num{20000} out of two million daily users select circuits that \name can speed up by \SI{120}{\milli\second} or more.

\paragraph{Contributions}
We propose \name, the first protocol to apply multi-hop overlay routing to an anonymous communication network.
\name is designed to improve performance while preserving the security guarantees of baseline Tor, preventing adversarial relays from gaining an advantage by participating in \name.
We evaluate \name using measured latencies from the live Tor network and show that \name can significantly improve tail latencies on the Tor network with minimal bandwidth overhead. 

In summary, this paper contributes: 
\begin{enumerate}
    \item 
        \name: a protocol for multi-hop overlay routing on Tor which reduces the latency experienced by Tor circuit traffic by \SI{122}{\milli\second} in the \num{99}\textsuperscript{th} percentile.
    \item 
        An evaluation of \name's performance at various levels of deployment based on measured latency between the thousand most popular Tor relays.
    \item 
        A security analysis of \name in the AnoA framework, demonstrating minimal impact to user anonymity.
\end{enumerate}

%% file: sections/background.tex

\noindent
Here, we provide background on Tor and multi-hop overlay routing, which we combine in \cref{sec:shortor} to design \name.

\subsection{Tor}

\begin{figure}[t]
     \centering
     \includegraphics[width=\linewidth]{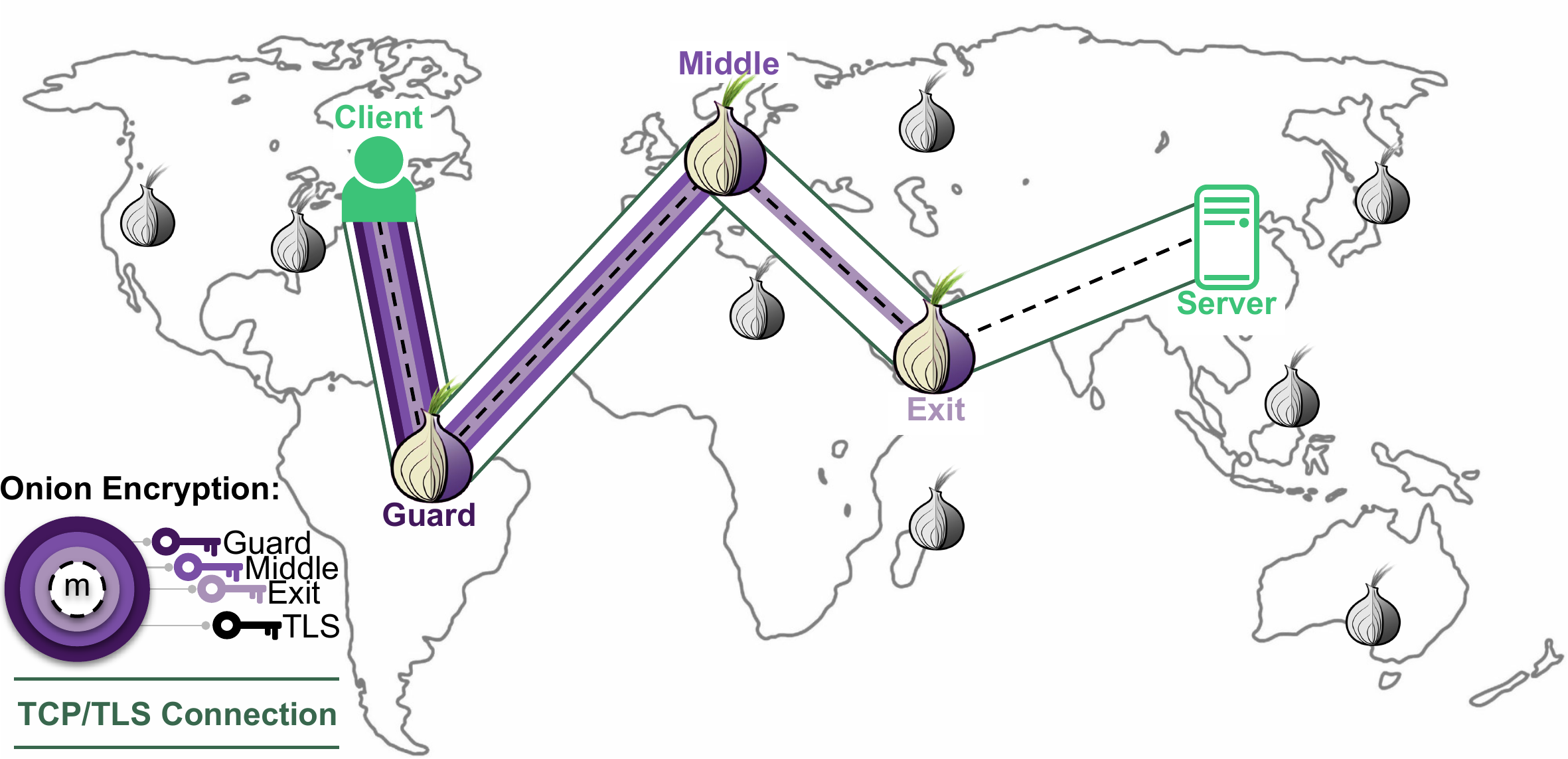}
     \caption{A Tor circuit between a client and server: Tor relays are represented by onions. The circuit is a series of connections between three relays carrying onion-encrypted Tor cells.}
     \label{fig:tor}
 \end{figure}

\noindent
The Onion Router (Tor) is a network for anonymous communication comprising approximately \num{7000}~\cite{tor-metrics-network-size} volunteer-run relays that carry user traffic.
We provide a brief overview of Tor's architecture and security guarantees.
For more details on Tor, see the Tor specification~\cite{tor-spec} or paper~\cite{tor-paper}.

\subsubsection{Onion Routing} Tor users send their traffic through the Tor network using \emph{onion routing}.
Rather than communicating directly with their destination, clients send their traffic through ``layers:'' encrypted connections to three (or more) Tor relays in sequence.
These relays form a \emph{circuit} and have fixed roles:
\begin{description}
\item[Guard\!] relays connect directly to the client and serve as an entry point into the Tor network,
\item[Exit\!] relays connect directly to the server and proxy communication on behalf of the client, and
\item[Middle\!] relays pass traffic between the Guard and Exit.
\end{description}
\Cref{fig:tor} shows a single Tor circuit between a client and server including the connections and layers of encryption involved in Tor's onion routing protocol.
The traffic flowing over a circuit is carried in fixed-size packets called \emph{cells} which are onion encrypted.
That is, cells have a layer of encryption for \emph{each} relay on the circuit.
Tor relays remove their layer of encryption when forwarding cells in the client-to-server direction and add their layer when returning the responses.
This ensures that only the Exit can remove the innermost onion layer, protecting the client’s privacy without requiring destination servers to handle onion encrypted data.

\subsubsection{Path/Circuit Selection} 
Path---or circuit---selection is the process by which Tor clients select the set of relays that will form their circuit.
This is a randomized process to ensure that the selection of relays is neither predictable nor correlated with the identities of either the client or server. 
It is, however, not \emph{uniformly} random as relays have highly variable capacities and not every relay can support the same volume of traffic. 
As such, path selection is weighted based on a relay’s available bandwidth, along with security considerations.

\subsubsection{Tor's Adversarial Model} 
\label{sec:background-tor-model}
Tor is intended to provide \textit{anonymity} to its users. 
Specifically, no adversary should be able to link the source and destination of any traffic stream across Tor.
Tor's threat model considers adversaries in the form of malicious relays as well as external observers such as users' internet service providers.
Anonymity in Tor is provided among all concurrent Tor users. 
While onion routing prevents any individual relay or localized network observer from directly linking a client to their destination, it does \textit{not} hide the fact that a client is connected to the Tor network in the first place. 
Similarly, onion routing alone does not hide which servers are the destination of Tor connections.
As such, both the \textit{volume} and \textit{diversity} of Tor users influence the degree of anonymity Tor is able to provide.
In a well-known example of this principle, the sender of a 2013 Harvard bomb threat  was identified despite their use of Tor because they were the \emph{only} client connecting to Tor from Harvard's campus at the time~\cite{tor-bomber-harvard}.

\subsubsection{Traffic Analysis}
\label{sec:traffic-analysis}
Traffic analysis attacks are a type of anonymity-compromising attack against Tor that identify features of encrypted traffic stream, such as packet interarrival times~\cite{li2018correlation}, to either: 1) recognize a previously observed stream~\cite{gilad2012spying,johnson2013users,nasr2018deepcorr,digestor}, linking it across Tor or, 2) observe a pattern corresponding to a website \textit{fingerprint} and infer the destination of traffic~\cite{jansen2018inside,fingerprinting-usenix,Bhat_2019,Rimmer_2018,sirinam2018deep}.
Both styles give the adversary an advantage in linking a client to their destination, compromising Tor's anonymity by making clients, servers, or client-server pairs more identifiable.
We give additional details on the capabilities of such adversaries and their impact on Tor in \Cref{sec:secanalysis}.

\begin{figure}[t]
     \centering
     \includegraphics[width=\linewidth]{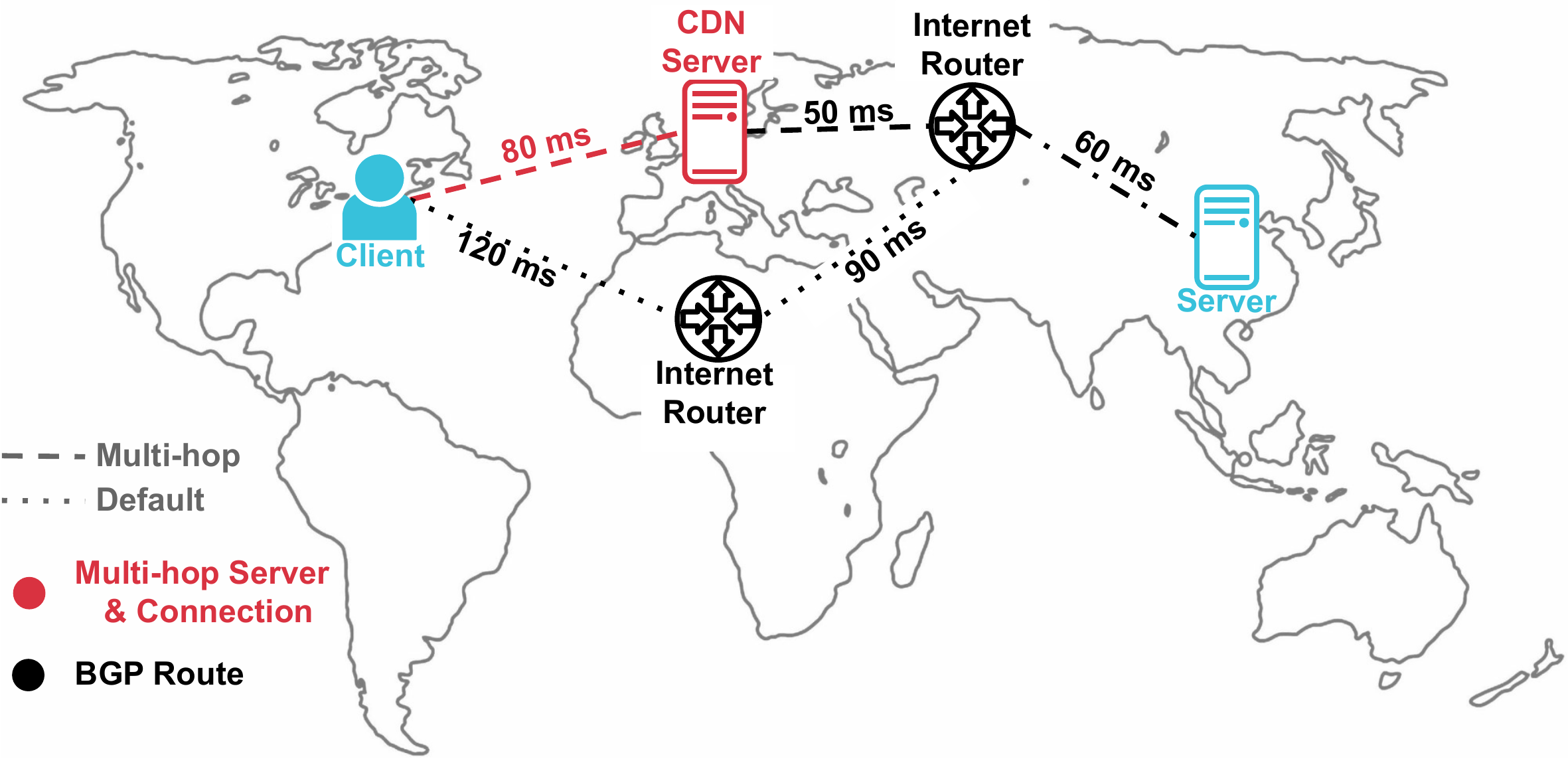}
     \caption{Multi-hop overlay routing as in a CDN: the client avoids a slow BGP route to the blue server by addressing data to the red CDN server, which then forwards the traffic.}
     \label{fig:overlay-routing}
 \end{figure}
\subsection{Multi-hop Overlay Routing}

Multi-hop overlay routing is a technique that introduces intermediate waypoints into the connection between a client and server for the purpose of altering the route their traffic takes across the internet.
There are many motivations for this technique---Tor's onion routing is itself an example of multi-hop overlay routing that provides anonymity by masking the direct relationship between client and destination server.
More commonly, CDNs route their traffic over a multi-hop overlay in order to reduce the latency of their connections as illustrated in \Cref{fig:overlay-routing} \cite{su2009drafting,bornstein2013optimal}.
Two such examples are Cloudflare’s Argo~\cite{argo} and Akamai’s SureRoute~\cite{sureroute}.
Rather than relying solely on the Border Gateway Protocol (BGP) to decide routes for their traffic, both Argo and SureRoute instead establish intermediate connections via their own servers.
By routing their traffic via these intermediate waypoints they are able to identify and use a route which may be faster than the route selected by BGP. This is possible because BGP is subject to routing policies based on business relationships, not solely on shortest paths~\cite{seetharaman2009interdomain}.

Importantly, this technique is an \emph{overlay}---it runs on top of standard BGP without modifying any of the underlying protocols.
This is achieved by establishing pairwise TCP connections between each of the intermediate points on the multi-hop overlay route rather than a single direct connection between the client and server.
As such, standard BGP handles the route used \textit{between} hops while the overlay protocol adjusts the ultimate path between client and server via the \textit{placement} of its waypoints.

%% file: sections/shortor.tex
We propose \name,  a protocol for reducing the latency of connections over Tor.
Like other such proposals, \name preferentially selects faster routes across the Tor network for client circuits.
In prior work, fast \emph{routes} across Tor are equivalent to fast Tor \emph{circuits}---Tor clients simply optimize for latency when selecting relays for their circuits. 

Instead, \name creates a \emph{multi-hop overlay} on top of the Tor protocol to improve latency as shown in \Cref{fig:shortor}. 
Rather than altering the circuit selection process to favor faster paths, \name changes the routing \emph{between} relays on existing circuits.
It does this by offering circuit relays the option to route their traffic through an additional Tor relay rather than directly to the next hop.
These intermediate hops, called \emph{via relays}, are chosen on-demand by the relays themselves instead of in advance by clients.
Via relays are \textbf{not} part of client circuits and do not participate in onion routing or encryption.

\begin{figure}[t]
     \centering
     \includegraphics[width=\linewidth]{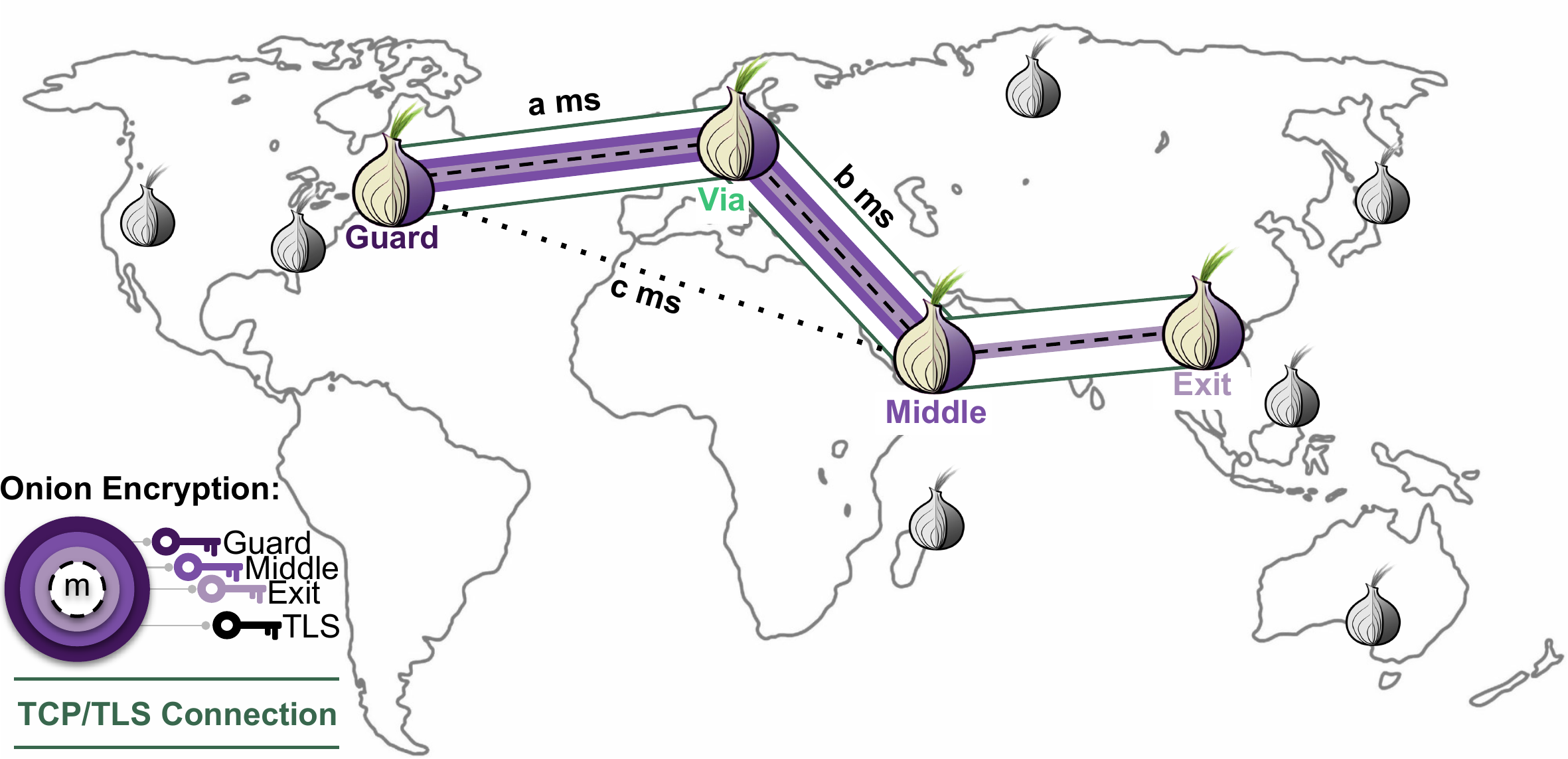}
     \caption{A Tor circuit routing using a via relay between the guard and middle.  A via relay will be used when the latency over the via (a+b ms) is less than that of the direct connection (c ms). The via does \textit{not} participate in onion routing. For clarity, and because Shortor only operates between Tor relays, the client and server are not shown.}
     \label{fig:shortor}
 \end{figure}

By routing as an overlay rather than altering circuit composition, \name avoids security pitfalls of prior works while still providing a substantial reduction in latency on the Tor network.
Directly optimizing for faster circuits, as in past proposals, has the unfortunate side effect of creating a correlation between the relays a client chooses and the client's location.
Via relays in \name are chosen \emph{only} based on the circuit relays and inherit their relationship to the client---if circuits are chosen independently, as in the default Tor circuit selection, then via relay choices leak no information about the client.
We discuss the implications of running \name with alternative circuit selection techniques in \cref{sec:secanalysis} and \cref{sec:discussion}.

\name's design gives it several advantages over proposals that modify circuit selection:
\begin{enumerate}
    \item \textbf{Security:} Routes in \name are independent of the client and destination.
    \item \textbf{Agility:} \name can modify its routes as needed, not just during circuit construction.
    \item \textbf{Compatibility:} \name operates with \emph{any} circuit selection algorithm, making it modular and compatible with future changes to the Tor protocol.
\end{enumerate}
While we describe \name in Tor-specific terms, we note that it applies to other distributed communication systems as well.

\subsection{Security Model}
\label{sec:shortor-security}
\name inherits Tor's adversarial model, as described in \cref{sec:background-tor-model}.
It is designed to preserve the same anonymity guarantees against an adversary identifying the sender or recipient of a traffic stream which we define more formally in \cref{sec:secanalysis}.
In particular, \name requires no modification to Tor's baseline circuit selection or encryption and preserves independence between circuit choice and the identities of the client and server.
However, \name \emph{does} necessarily change the number and distribution of relays that may see a given traffic stream, which could potentially be exploited by an adversarial Tor relay to deanonymize a larger share of Tor traffic. 
We formally consider the anonymity impact of \name in our security analysis (\Cref{sec:secanalysis}).

\input{sections/shortor_protocol}

\subsection{\name Protocol}
\label{sec:shortor-proto}
\name introduces only one additional step into Tor's routing procedure.
Rather than forwarding cells solely along previously established circuits, relays establish transient alternate routes between themselves and the next hop on their circuits.
These alternate routes forward traffic via an additional Tor relay rather than sending it directly to the next relay on the circuit.
As such, we refer to the intermediate hops between circuit relays as \emph{via relays}, the connection between a circuit relay and a via relay as a \emph{via connection}, and the communications over this connection as \emph{via traffic}.

Note that the `circuit' and `via' modifiers denote different roles a relay may play in \name, but do \emph{not} correspond to different physical entities. 
A via relay is simply a regular Tor relay that has been chosen as an intermediate hop for some circuit rather than as part of the circuit itself. 
Any relay in Tor can act as both a circuit and a via relay simultaneously for different traffic streams.

\subsubsection{\name Protocol Stages}
\label{sec:shortor-proto-stages}
The \name protocol proceeds in several stages. 
On an ongoing basis, relays take measurements of their round-trip latencies with other relays ($\algname{Latencies.Update}()$, \cref{proto:latency-measurements}).
Circuit relays use these measurements to choose candidate via relays for outgoing traffic ($\algname{\name.ChooseVia}()$, \cref{proto:shortor}).
They perform a ``data race'' to choose the fastest path ($\algname{Race}$, \cref{proto:data-race}).
If a route with a via relay is faster than the default path, the circuit relay updates its routing table.
In the steady state, the circuit relay handles traffic for its circuits as usual, but directs it to the via relay rather than to the next circuit hop.

\paragraph{Establishment} When establishing a connection for a given circuit, relays on that circuit will run $\algname{Latencies.ViaFor}()$ (\cref{proto:latency-measurements}) to obtain a shortlist of potential vias.
These vias are those that have recently been observed to provide the largest latency improvements over the default path between this relay and the next hop on its circuit.
The circuit relay then performs a data race over each of the candidate vias ($\algname{Race.Run}()$, \cref{proto:data-race}). 
The finish line of this race is the next relay on the relevant circuit which can report to the starting relay which of the data race cells arrived first. 
We provide details on the selection of candidate via relays in \cref{sec:latency-measurements}. 

\paragraph{Routing} While establishing a via connection, both the circuit and via relays must update their routing tables: circuit relays note which via to send cells to, while vias record which circuit relay should receive their forwarded traffic.
To do this, we simply introduce new fields in Tor cell headers and routing tables, described in \Cref{fig:via-cell}.
These allow relays to recognize traffic streams and route them to the correct next hop. 

\begin{figure}[t]
    \centering
    \includegraphics[width=\linewidth]{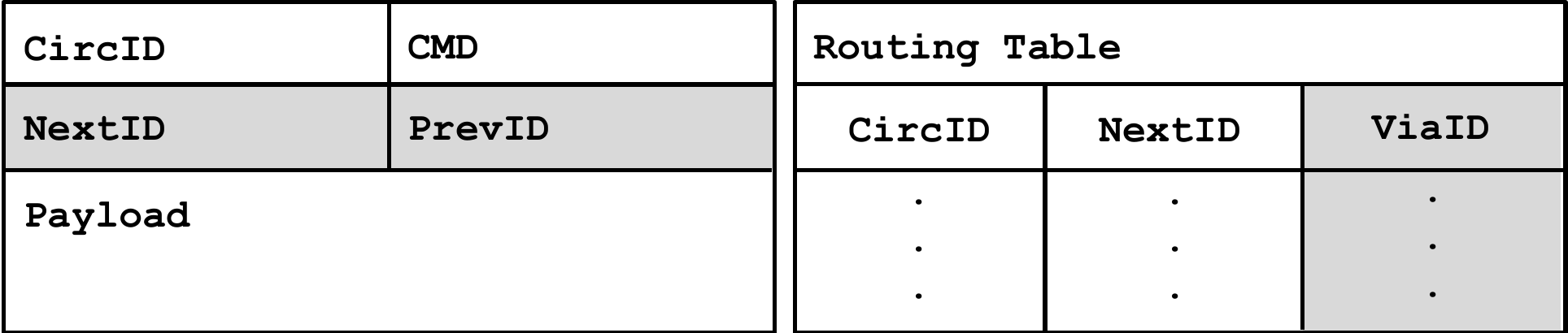}
    \setlength{\belowcaptionskip}{-10pt}
    \caption{Via Cell and Routing Table (fields added to baseline Tor highlighted in grey): Via cells contain additional header fields with routing info. This information is used by circuit relays to populate the routing table with which via (if any) each circuit should be routed through, while via relays use it to record where to forward cells from a given circuit.}
    \label{fig:via-cell}
 \end{figure}

\paragraph{Steady-State} Traffic over via connections that have already been established is handled similarly to regular Tor traffic.
Via relays simply forward the received traffic stream according to their routing table for the circuit.
As via relays are not part of client circuits, they do \emph{not} perform onion decryption/encryption prior to forwarding cells.
Circuit relays also function as in baseline Tor except in cases where their routing table for a circuit contains a via relay.
Then, the relay will alter the header on cells for these circuits as shown in \Cref{fig:via-cell} and send the cells to the indicated via rather than to the next relay on the circuit.
Periodically, relays on a circuit can repeat the data race to determine whether a given via is still the fastest option based current network status.

\subsubsection{Latency Measurements}
\label{sec:latency-measurements}

\begin{figure}
\begin{protocol}{\algname{Latencies}}{latency-measurements}
    \protocolheader{All Relays}\\
    \algcomment{Parameter: $\ell$, how many routes to keep.} \\
    \algcomment{Parameter: IDs of all $n$ active Tor relays: $\variable{Tor} = \{\variable{id}_1, \dots, \variable{id}_n\}$.} \\
    \algcomment{State: table $\variable{RTTs}$: ping times to \emph{each} other relay.} \\
    \algcomment{State: table $\variable{NextHop}$: for each relay, the top $\ell$ candidate vias $(\variable{id}, \variable{rtt})$.}
    
    \begin{itemize}[leftmargin=10pt]
        \item {$\algname{Latencies.Update}()$}\\
        \algcomment{Keep $\variable{RTTs}$ and $\variable{NextHop}$ tables up-to-date.}\\
        \algcomment{Run periodically (once per day).}
        \begin{algenumerate}[leftmargin=10pt]
            \item \For {~$\variable{id} \in \variable{Tor}$}:
            \begin{algthenenumerate}[label=1.\arabic*:]
            	\item Ping relay $\variable{id}$ to estimate round-trip time (RTT).
                \item Set $\variable{RTTs}[\variable{ID}]$ to estimated value.
                \item Remove $(\_, \variable{rtt}) \in \variable{NextHop}[\variable{id}]$ with $\variable{rtt} \ge \variable{RTTs}[\variable{id}]$.
                \item $\variable{RTTs}_{\variable{id}} \gets \algname{Latencies.RTTs}()$ (remote call to relay $\variable{id}$).
                \item \For $\variable{via} \in \tor$
                \begin{algthenenumerate}[label=1.5.\arabic*:]
                    \item $\variable{rtt} \gets \variable{RTTs}[\variable{via}] + \variable{RTTs}_{\variable{id}}[\variable{via}]$.
                    \item \If $\variable{rtt} \ge \variable{RTTs}[\variable{id}]$ \Then \textbf{continue}  \hfill\algcomment{no speedup.}
                    \item Add $(\variable{via}, \variable{rtt})$ to $\variable{NextHop}[\variable{id}]$, keeping fastest $\ell$ entries.
                \end{algthenenumerate}
            \end{algthenenumerate}
        \end{algenumerate}

        \medskip
        \item {$\algname{Latencies.RTTs}()$}
        \begin{algenumerate}
            \item Output $\variable{RTTs}$.
        \end{algenumerate}

        \medskip
        \item {$\algname{Latencies.ViasFor}(\variable{id})$}
        \begin{algenumerate}
         	\item Output $\variable{Routes}[\variable{id}]$. \hfill\algcomment{up to $\ell$ candidate via relays.}
        \end{algenumerate}
    \end{itemize}
\end{protocol}
\end{figure}

\name relies on two forms of latency measurements (1) an up-to-date table of probable via candidates for each relay pair (\cref{proto:latency-measurements}) and (2) the data race that determines the fastest of the candidates (\cref{proto:data-race}).

\subsubsection{Pairwise Latency}
\label{sec:pairwise-latency-alg}
\name requires latency measurements between Tor relays to narrow down the set of potential via relay options for a circuit.
In $\algname{Latencies.Update}()$ (\cref{proto:latency-measurements}), each Tor relay collects this latency information as needed, distributing the involved storage, computation, and network load across the Tor network.
This is in contrast to the centralized measurement methodology we use to evaluate \name in \cref{sec:measurements} which, while useful for this work, would not meet the performance needs of the live \name protocol.

To participate in the distributed latency measurements of \cref{proto:latency-measurements}, each relay maintains their estimated round-trip latency to every other relay along with a list of ``candidate'' via relays.
The candidates are computed by each relay using the round-trip latency tables for itself and for the destination relay based on latencies provided by the \emph{destination}.
\name uses latencies reported from the destination for security reasons: an honest destination will not recommend a dishonest via relay disproportionately often, while a dishonest destination was \textit{already} on the circuit and gains nothing by lying.
The list of candidate vias is used to inform the data race which will select the fastest via from the list at the time of the race. 

\subsubsection{Data Race}
\label{sec:data-race}
Directly choosing via connections based on measured latencies has several potential drawbacks.
First, the measured latencies are round-trip, while network paths are directional: the fastest path from relay A to relay B might be different from the fastest path from relay B to relay A.
Timestamping at the destination halfway through the round trip \emph{does not} address this issue, as it becomes impossible to distinguish between imperfect clock synchronization and path asymmetry.
Second, latencies change in real-time in response to network conditions, like congestion at relays or on internet links.
Third, latencies might be inaccurate due to measurement errors or even misreporting by malicious relays; we must take care to prevent such relays from seeing disproportionate amounts of traffic.
As such, measured latencies alone are insufficient.

Instead circuit relays choose the fastest via using a ``data race:'' sending packets along different routes to see which arrives at the destination first ($\algname{Race.Run}()$, \cref{proto:data-race}).  
The starting relay \emph{simultaneously} sends a copy of a data race cell to each prospective via relay, and one copy directly to the destination.
The destination relay, which is the next hop on the circuit, responds \emph{only} to the first of these cells to arrive.

Data races are \emph{directional}---relays can identify the fastest path in each direction separately.
Additionally malicious via relays cannot report lower latencies to artificially increase their odds of being selected.
Data race cells are not forgeable by the via, so the via must wait to receive the cell from the source circuit relay before delivering it to the destination relay. 
Thus, the via cannot artificially reduce its perceived latency below the true time it takes to forward the cell.

\begin{figure}
\begin{protocol}{$\algname{Race}$}{data-race}
\protocolheader{Circuit Relay:}\\
\algcomment{Find the fastest via relay for reaching the destination relay.}\\
\algcomment{Parameter: $\variable{myId}$, Tor ID of this relay.}\\
\algcomment{State: $\variable{Seen}$, a set of circuit IDs for which this relay has seen data race packets.}
\begin{itemize}[leftmargin=10pt]
\item {$\algname{Race.Run}(\variable{vias}, \variable{circ}, \variable{dst})$}\\
\alginput{Candidate vias $\variable{vias} = \{v_1, \dots, v_\ell\}$, destination relay \variable{dst}}\\
\algoutput{Fastest via $v$ if one exists; $\bot$ otherwise.}
\begin{algenumerate}
   \item Create data race cell \variable{cell} with fields
     	    $\variable{cmd} = \constant{RACE}$,
     	    $\variable{prev} = \variable{self}$,
     	    $\variable{next} = \variable{dst}$, and
      	    $\variable{circ} = \variable{circ}$.
 	\item \For {~$\variable{via} \in \variable{vias}$}: send $\variable{cell}$ to $\variable{via}$.
 	\item Send data race cell directly to \variable{dst}.
 	\item $\variable{resp} \gets$ response from \variable{dst}.
 	\item Output $\variable{resp}.\variable{via}$. \hfill\algcomment{May be $\bot$ if no via provides speedup.} 
\end{algenumerate}

\medskip
\item {$\algname{Race.Respond}(\variable{cell})$}\\
\alginput{Data race cell from source relay (\variable{sender}).}
\begin{algenumerate}
 	\item \If $\variable{cell.circ} \in \variable{Seen}$ \Then drop cell and return.
 	\item Add $\variable{cell.circ}$ to $\variable{Seen}$.
 	\item \If $\variable{cell.prev} = \variable{sender}$ \Then $\variable{via} \gets \bot$
 	\item \Else $\variable{via} \gets \variable{sender}$.
 	\item Send response $\variable{resp}$ to \variable{cell.prev} with $\variable{resp.via} = \variable{via}$.
\end{algenumerate}
\end{itemize}

\bigskip
\protocolheader{Via Relay:}\\
\algcomment{Via relays update their routing tables to forward traffic  on a stream, \algcommentbreak provided sufficient resources are available to do so.}\\
\algcomment{State: $\variable{Routes}$, the routing table for each circuit.}
\begin{itemize}[leftmargin=10pt]
\item {$\algname{Race.ViaForward}(\variable{cell})$}\\
\alginput{Data race cell $\variable{cell}$ with $\variable{cell.cmd} = \constant{RACE}$.}
\begin{algenumerate}
    \item \If under heavy load \Then drop cell and return.
    \item Add \variable{cell.prev} and $\variable{cell.next}$ to $\variable{Routes}[\variable{cell.circ}]$.
    \item Forward the cell to relay $\variable{cell.next}$.
\end{algenumerate}
\end{itemize}
\end{protocol}

\end{figure}

\subsubsection{Avoiding Traffic Loops}
We define a \textit{loop} to occur when the same traffic stream passes through a relay more than once.
This is an issue as such relays could utilize traffic correlation to identify the previously seen traffic stream, thus learning a larger portion of its path through Tor than they should have been privy to.
Tor only builds circuits using distinct, unrelated relays to ensure that circuits contain no loops.
However, because \name selects via relays separately from the circuit selection process, care must be taken to avoid loops.

In order to provide the same guarantee as Tor, we require that \name is applied only to circuits of length exactly three (the default in Tor) and that only a single via is used between any pair of relays. 
This ensures that the middle relay of a circuit is capable of observing \textit{all} vias on that circuit and enforcing the same guarantees as for circuit relays.
That is, in either direction of a circuit, the middle relay will not choose to use a via that is already in use for the prior hop or is related to a relay on the circuit.
We elaborate on security in \cref{sec:secanalysis}, but note here that a malicious middle relay gains no advantage by failing to enforce this guarantee, as it already knows the identities of both the guard and exit relays and does not need to correlate traffic across the via to get this information.

\subsubsection{Stability}
\label{sec:stability}
\name's distributed via selection protocol must avoid \emph{oscillations} where circuit relays swap back and forth between vias.
As an example, without appropriate precautions, a cycle could form where traffic streams dropped from an overwhelmed via all divert to the same alternate via, subsequently overwhelming that via and causing the streams to revert to the original choice, and so on. 
We note that this situation is not prohibitive: CDNs use a similar overlay routing technique in practice, carrying substantial portions of internet traffic, without such stability problems.

We mitigate the risk of this situation in \name through backoff and capacity parameters in the data race.
We note that races are an integral component of the \name protocol and are conducted to identify faster routes for traffic, not for stability purposes. 
These backoff and capacity parameters simply ensure that races avoid oscillations between vias.

Circuits never attempt to send traffic through a via without first conducting a race (though they may fall back to their direct path at any point).
First, vias without available capacity will drop data race cells, preventing them from being selected at the small cost of processing a single packet.
Second, upon being dropped from a via path, circuit relays will apply randomized exponential backoff and will not include that via in data races again until a set period of time has passed.
The exact parameters for via capacity and backoff timing are network-dependent and may evolve based on the current state of the Tor network.
However, because \name is an optional performance-enhancement, these values can  initially be set conservatively and then decreased adaptively.

\subsection{Integration with Tor}
\label{sec:tor-integration}
While the \name approach has potential applications to other networks, we designed and evaluated \name's protocol to integrate with Tor specifically. 
Maintaining Tor's existing security guarantees is one main focus of this design and informed the structure of our data races and avoidance of loops.
However, successful integration with Tor also requires that \name be \textit{deployable}.
In this section, we discuss the components of \name that are most relevant Tor deployment, including support for load balancing and fairness, required modification to Tor relays, and incremental deployment.

\subsubsection{Load Balancing \& Fairness}
\label{sec:fairness}
Fairness to circuit traffic and load balancing are both necessary to ensure that \name does not inadvertently increase latency for some circuits as a consequence of reducing it on others.
This could happen if via traffic was allowed to consume more resources than a relay had available to spare, resulting in increased processing times or congestion at the relay.
\name provides both fairness and load balancing through the same mechanism: prioritizing circuit traffic over via traffic.
Tor already recognizes different traffic priorities---web browsing is prioritized over large file downloads~\cite{tor-spec}. 
We extend this to ensure that relays will preferentially schedule traffic from circuit queues over via queues (\Cref{fig:tor-queuing}).

\begin{figure}[t]
     \centering
     \includegraphics[width=\linewidth]{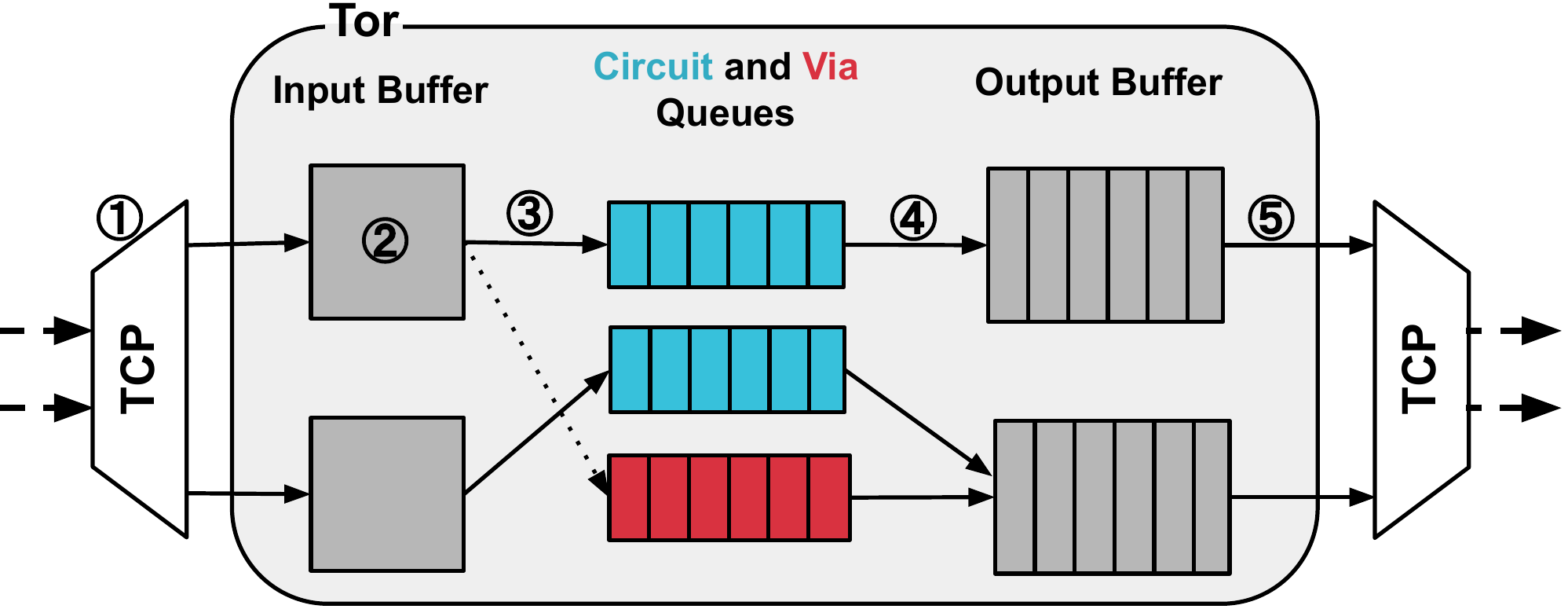}
     \caption{Tor relay with circuit and via traffic (queuing architecture unmodified from baseline Tor).
     \circnum{1} Via and circuit traffic are multiplexed on a TCP connection entering the Tor relay. \circnum{2} The TLS layer is decrypted and circuit cells are onion encrypted/decrypted. \circnum{3} Circuit and via cells are sent to their individual queues. \circnum{4} Cells are scheduled for release to the output buffer based on priority order. \circnum{5} The contents of the output buffer are encrypted using TLS, then sent to the kernel for transit over a TCP connection to the next hop.}
     \label{fig:tor-queuing}
 \end{figure}

Circuits will select vias that have lower latency than the default path, \emph{including} the transit time through the via itself. 
This is very important as relay congestion and the associated queuing delays are a primary source of latency in Tor~\cite{jansen2018kist,jansen2014never,jansen2017tor}.
Congestion at a via will appear naturally during the data race in the form of increased latency or could be indicated explicitly by dropping race packets.
In addition, vias are \textit{transitory} and can be dropped or swapped at will with minimal cost compared to that of circuit construction/teardown.

As such, \name ensures that: (1) circuit traffic on a relay is never delayed by via traffic and (2) load from via traffic is distributed only across relays with available capacity.

\subsubsection{Tor Modifications}
\name's primary modification to Tor is the introduction of data races, all other components are simple extensions of Tor's existing mechanisms for routing and prioritizing circuits. 
To support \name{}, Tor relays (though not clients) require additional protocol messages, a new data path for via traffic, and state for managing via traffic.
The protocol requires new cell headers for data races, ping, and via traffic.

Specifically, via traffic needs a new priority level lower than circuit traffic (optionally, this level can be higher than that of bulk download traffic, such as torrenting, which is currently of lower priority than circuit traffic)~\cite{tor-spec}.
Incoming via traffic needs a new data path that bypasses onion encryption and decryption.
Relays must also handle ping and data race traffic as specified in \Cref{proto:latency-measurements} and \Cref{proto:data-race}.
Finally, relays must hold two additional pieces of state: first, a new field in the routing table to indicate the via (if any) for each circuit; second, the list of candidate via nodes for each possible next hop (see \cref{sec:latency-measurements} for details).

These modifications are relatively minor, do not touch Tor's onion encryption layer, and represent an optional overlay on Tor's routing.
We discuss more details of required modifications in \Cref{appendix:tor-modifications}, but note here that the high up-front cost of integrating and deploying modifications to the Tor protocol was a large factor in the ultimate design of \name.
This consideration motivated \name's construction as an extension to Tor's existing architecture that operates largely separately from the baseline protocol.

Furthermore, \name's modifications are configurable, trivially backwards compatible,\footnote{Relays lacking support for \name simply route as usual without any vias.} and support incremental deployment.
This allows relay operators to choose whether to support \name and how much capacity to dedicate to the protocol. 
As shown in \cref{sec:eval-incremental}, \name can substantially reduce tail latencies even with relatively low support.
This is important, as it minimizes the risk of up front development efforts being wasted due to slow deployment.

\subsubsection{Incremental Deployment}
\label{sec:incremental}
Tor relays are volunteer run and notoriously slow to update~\cite{tor-metrics}.
As such, any proposal that requires support of all---or even a majority of---Tor relays is unlikely to be effective. \name is incrementally deployable and improves the latency of any Tor circuits that meet the following two requirements: (1) two adjacent circuit relays support \name and (2) some other relay supporting \name provides a faster path between the two circuit relays.
Because Tor does not select its relays with uniform probability, a small set of popular relays could meet these conditions for many circuits without support from the rest of the network.
We demonstrate this concretely in \cref{sec:eval-incremental}.

Security is another important consideration---incremental deployment inherently creates differences between Tor clients or relays that have adopted a modification and those who have not.
This has been an issue for client side proposals as anonymity in Tor relies on all clients behaving \emph{uniformly}~\cite{rochet2020claps,mitseva2020security,mator,anoa-mators-thesis,anoa,geddes2013how,wacek2013empirical}.

\name avoids this issue entirely as it is a fully server-side protocol that does not require participation from, or modify the behavior of, Tor clients in any way.
So, while \name is an observable modification to the Tor protocol,\footnote{Both adversarial relays and network adversaries can likely detect when traffic is routed using \name as opposed to baseline Tor.} it is in no way correlated to client identity.
As such, support for \name \emph{cannot} be used to distinguish between clients.
In fact, Tor clients should \emph{not} try to preferentially select relays with support for \name.
While this would improve their latency, it would also differentiate them from Tor clients following Tor's baseline circuit selection algorithm, reducing their anonymity.

%% file: sections/shortor_protocol.tex
\begin{figure}[t]
\begin{protocol}{\name}{shortor}
\medskip
\protocolheader{Circuit Relay}

\algcomment{Circuit relays conduct data races to determine if 
a suitable via exists
\algcommentbreak between themselves and the next relay on the circuit. }\\
\algcomment{Parameter: $\variable{self}$, Tor ID for this relay.}\\
\algcomment{State: $\variable{Routes}$, the routing table for each circuit.}

\begin{itemize}[leftmargin=10pt]
\item 
{$\algname{\name.ChooseVia}(\variable{circ}, \variable{dst}, \variable{candidates})$}
\begin{algenumerate}
 	\item $\variable{via} \gets \algname{Race.Run}(\variable{candidates}, \variable{circ}, \variable{dst})$ \algcomment{(\cref{proto:data-race}).}
 	\item \If $\variable{via} = \bot$, \Then return. \hfill\algcomment{no via faster than default route.}
 	\item $\variable{Routes}[\variable{circ}].\variable{via} \gets \variable{via}$.
\end{algenumerate}

\medskip
\item 
{\algname{\name.HandleTraffic}(\variable{cell})}
\begin{algenumerate}
  \item \If $\variable{cell.cmd} = \constant{VIA}$ \Then return $\algname{\name.HandleVia}(\variable{cell})$. 
  \item $\variable{candidates} \gets \algname{Latencies.ViasFor}(\variable{cell.next})$. \algcomment{(\cref{proto:latency-measurements})}
  \item \If $\variable{Routes}[\variable{cell.circ}] = \bot$ \Then \hfill\algcomment{no routing table entry.}
  \begin{algthenenumerate}[label=3.\arabic*:]
        \item $\algname{\name.ChooseVia}(\variable{cell.circ}, \variable{cell.next}, \variable{candidates})$. 
  \end{algthenenumerate}
  \item $\variable{via} \gets \variable{Routes}[\variable{cell.circ}].\variable{via}$.
  \item \If $\variable{via} = \bot$  \Then proceed with default cell routing and return.
  \item Set $\variable{cell.cmd} = \constant{VIA}$ and $\variable{cell.prev} = \variable{self}$.
  \item Send $\variable{cell}$ to relay $\variable{via}$.
  \item \If no response from $\variable{via}$ \Then 
        \begin{algthenenumerate}[label=8.\arabic*:]
            \item $\variable{candidates} \gets \variable{candidates} \setminus \{\variable{via}\}$
            \item $\algname{\name.ChooseVia}(\variable{cell.circ}, \variable{cell.next}, \variable{candidates})$. 
            \item $\algname{\name.HandleTraffic}(\variable{cell})$
        \end{algthenenumerate}
\end{algenumerate}
\end{itemize}

\bigskip
\protocolheader{Via Relay} \\
\algcomment{Via relays forward cells between circuit relays. 
\algcommentbreak Via relays do not perform onion decryption and only forward traffic \algcommentbreak if they have the available resources (i.e., bandwidth). }

\begin{itemize}[leftmargin=10pt]

\item 
{$\algname{\name.HandleVia}(\variable{cell})$}
\begin{algenumerate}
   \item $\variable{route} \gets \variable{Routes}[\variable{cell.circ}]$
   \item \If under heavy load or $\variable{route} = \bot$ \Then drop cell and return.
   \item Forward cell to $\variable{route.next}$.
   \item Forward response from $\variable{route.next}$ to $\variable{route.prev}$.
\end{algenumerate}
\end{itemize}
\end{protocol}
\end{figure}

%% file: sections/measurements.tex
 \begin{figure*}[t]
     \centering
     \includegraphics[width=\linewidth]{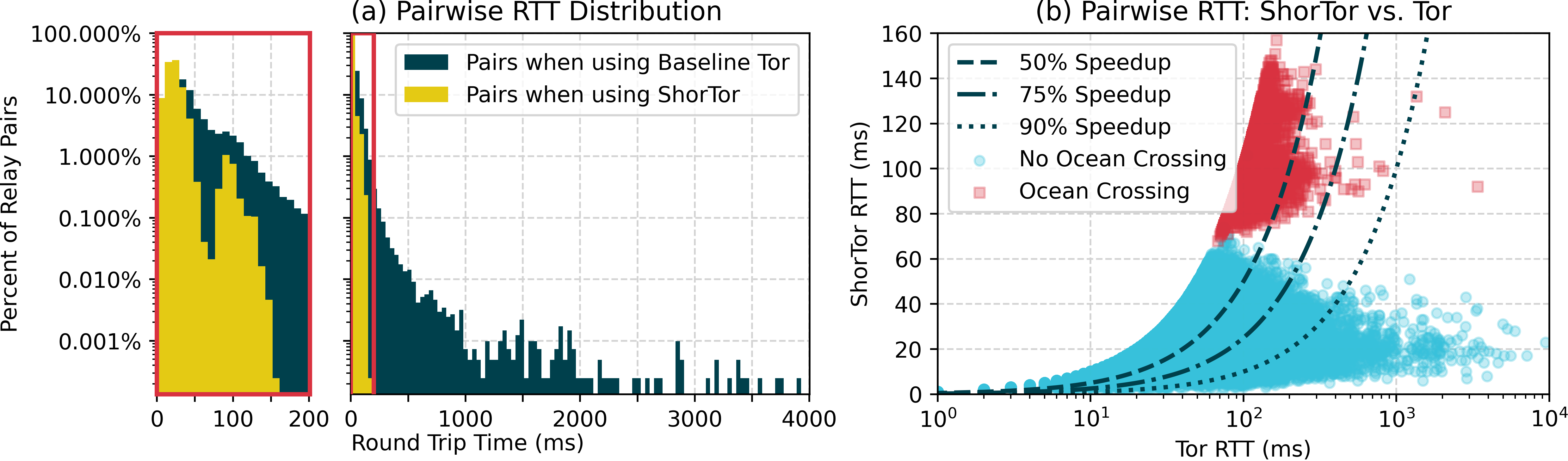}
     \caption{ \textit{(a) Left:} Round-trip times (RTTs) measured between Tor relays vs. RTTs of those \emph{same} pairs when using \name as a percentage of all \num{406074} pairs. We show an expanded view of the first \SI{200}{\milli\second} on the left. For readability, this figure omits 10 pairs with RTTs between \SI{4019}{\milli\second} and \SI{9415}{\milli\second}. \textit{(b) Right:} Relative latency of each relay pair when routing via \name vs. the default route. \SI{25.4}{\percent} of pairs have a latency reduction of at least \SI{50}{\percent}, \SI{9.4}{\percent} at least \SI{75}{\percent}, and \SI{1.2}{\percent} at least \SI{90}{\percent}.}
     \label{fig:rtt-dist}
 \end{figure*}

We evaluate \name using a dataset of approximately \num{400000} latency measurements we collected from the live Tor network over the course of \num{42} days during summer 2021.
Our measurements allow us to compare the direct latency between relays to the latency when routing through an intermediate hop, as in \name.

Using measured latencies allows us to avoid relying on simulated or approximate data.
While simulations can be a useful tool, prior work~\cite{schlinker2019peering} has shown that routing protocols are best evaluated using live internet paths rather than through a simulation with, necessarily, reduced scale and complexity.

We evaluate the performance of \name in terms of its direct impact on the latency between pairs of Tor relays as well as its ability to reduce the latency of Tor circuits.
Evaluating on circuits as well as pairs allows us to account for the relative popularity of relays and more closely model the expected reduction in latency \name can provide to Tor's end users.

\subsection{Measurement Methodology}
\label{sec:methodology}

\subsubsection{Ting}
\label{sec:ting}
For our measurements, we adapt the Ting method of \citet{cangialosi2015ting} for estimating latencies between Tor relays.
Ting creates a set of three circuits involving \emph{observers}, which are Tor relays run solely for the purpose of obtaining latency measurements.
Specifically, to obtain the latency between two Tor relays, $\relay_A$ and $\relay_B$, we run two observer relays $\obs_1$ and $\obs_2$ along with a measurement client on the same physical machine.
Once each circuit is established, the measurement client ``pings'' itself through the circuit to estimate round-trip latencies for the following circuits:
\begin{enumerate}
\item {\small $\rtt_{AB} = \textrm{RTT}\left(\obs_1 \to \relay_A \to \relay_B\to \obs_2\right)$}
\item {\small $\rtt_{A} = \textrm{RTT}\left(\obs_1 \to \relay_A\to \obs_2\right)$}
\item {\small $\rtt_{B} = \textrm{RTT}\left(\obs_1 \to \relay_B \to \obs_2\right)$}
\end{enumerate}
With these, we estimate the round-trip time between $\relay_A$ and $\relay_B$ as \(\rtt_{AB} - \frac{1}{2} (\rtt_A + \rtt_B)\).
This approximates the RTT including the forwarding delay at both $\relay_A$ and $\relay_B$ as forwarding is inherently a component of the latency experienced by via traffic.
We repeat this process in order to find the minimum observed latency between $\relay_A$ and $\relay_B$: in our observations, after \num{10} iterations, \SI{95.5}{\percent} of circuits are within \SI{5}{\percent} or \SI{5}{\milli\second} of the minimum observed in \num{100} samples. 
\subsubsection{Directional Latencies}
\label{sec:ting-mod}
The Ting protocol does not account for \textit{directional} latencies where the outgoing latency between two nodes may not be equivalent to that of the return trip.
Specifically, the method for computing $\rtt_{AB}$ described above assumes that $\textrm{latency}\left(\obs_1 \to \relay_{X}\right) \approx \textrm{latency}\left(\relay_X\to \obs_2 \right)$. 
To detect asymmetry in our RTT measurements we include a timestamp in our measurements halfway through the round-trip ``ping'' (all timestamps are with respect to the same clock).
In our dataset (\cref{sec:lat-data}), the median asymmetry was \SI{2.4}{\percent} and only \SI{0.2}{\percent} of measurements had an asymmetry of $2 \times$ or greater.
Importantly, asymmetric RTTs impact \textbf{only} our evaluation as, when deployed, \name{} naturally accounts for directional latencies using data races (\cref{sec:data-race}).

\subsubsection{Infrastructure}
\label{sec:infrastructure}
To collect latency measurements at scale, we adapted the Ting protocol to support parallel measurements across multiple machines.
Our larger scale also required changes to respect a safe maximum load on the Tor network (see \cref{sec:ethics}): we impose a global maximum limit on concurrent measurements and spread measurements of individual relays across time.
Our infrastructure also compensates for the high churn in the Tor network (\SI{13}{\percent} of relays we observed were online less than half the time) by enqueueing measurement jobs based on the currently online relays, with automated retries. 
We handled hardware and power failures using a fault-tolerant system design: we separated data persistence, measurement planning, and the measurements themselves.

We deployed to a private OpenStack~\cite{sefraoui2012openstack} cloud, but provide a Terraform~\cite{terraform} template supporting any provider.
Our open-source~\cite{source-code-for-shortor} measurement infrastructure is approximately \num{3300} lines of code, consisting primarily of Python and shell scripts.

\subsubsection{Geographic Location of Relays}
\label{sec:location}
We obtain country codes for the relays in our dataset using the GeoIP database \cite{geoip}, which is also used by Tor in practice.
However, GeoIP locations are not guaranteed to be \SI{100}{\percent} accurate~\cite{geo-ip-bad-1,geo-ip-bad-2,geo-ip-challenges-tor}.
Indeed, upon careful inspection, we observed a handful of relay pairs with \emph{physically impossible} RTTs for their purported locations. 
All of these pairs involved the same twelve relays allegedly located in the US.
We determined that these twelve relays have \textit{higher} average RTTs inside the Americas than they do to relays located in other regions.
Because of this, all location-related figures (\Cref{fig:rtt-dist}(b), \Cref{fig:mator}, \Cref{fig:network-share}) exclude these relays as, while we are confident that their reported location is incorrect, we cannot accurately determine their true location. 

\subsubsection{Ethics and Safety}
\label{sec:ethics}
We designed our measurement process to minimize impact on Tor users and relay operators and to comply with security best practices for Tor.
To this end, we submitted a proposal to the Tor Research Safety Board~\cite{tor-research-safety-board} for review prior to measurement and adhered to their recommendations.
We also received an IRB exemption from each author's institution for this work.

Collecting our data required us to run several live Tor relays.
These relays recorded \textbf{only} our measurement traffic---at no point did we observe or record any information about any traffic from Tor users.
We also minimized the likelihood of a user choosing our relays for their circuits by advertising the minimum allowed bandwidth of \SI{80}{\kibi\byte\per\second}~\cite{tor-spec}.

Our measurement collection was spread over 42 days to reduce concurrent load, including a limit on simultaneous measurements (detailed in \Cref{sec:infrastructure}).
We also notified a Tor relay operator mailing list and allowed operators to opt out of our measurements; we excluded four such relays.

In light of recent work by \citet{Schnitzler2021WeBuilt} on the security implications of fine-grained latency measurements for Tor circuits, we have not published our full latency dataset.
However, we will share this data with researchers upon request and are in communication with our reviewers from the Tor Research Safety Board about safely releasing it in the future.

\subsubsection{Latency Dataset}
\label{sec:lat-data}
In this work, we directly measure pairwise latencies within Tor rather than relying on outside estimates.
We focus our measurements on the \num{1000} most popular Tor relays (by consensus weight) for two main reasons:
\begin{enumerate}
    \item 
        Measuring all \num{36325026} possible pairs of the \num{8524} Tor relays we observed was intractable for this work.
    \item 
        These popular relays are present on over \SI{75}{\percent} of circuits~\cite{greubel2020quantifying} and thus can provide disproportionate utility.
\end{enumerate}

\noindent
Our dataset contains \num{406074} measured latencies or \SI{81.3}{\percent} of all pairs of the \num{1000} most popular relays.\footnote{Missing measurements are largely due to churn in the Tor network causing relay pairs to not be live simultaneously.}

\subsection{Applying \name to Relay Pairs}
\label{sec:pairs-perf}

We begin evaluating \name by comparing the potential latency between our set of relay pairs when routing via \name to our measured latencies observed using Tor's default routing. 
\Cref{fig:rtt-dist}(a) shows the relative frequency of RTTs experienced by pairs of Tor relays using \name and when routing normally while \Cref{fig:rtt-dist}(b) focuses on the relationship between default RTT and \name RTT for each relay pair. 
Using \name, all of Tor's high tail latencies were resolved: \name sees a maximum absolute RTT of \SI{157}{\milli\second}, while \SI{0.09}{\percent} of pairs in Tor had RTTs of over half a second.
In other words, the 99.9\textsuperscript{th} percentile of relay pairs see a reduction in RTT from \SI{487}{\milli\second} in Tor to \SI{125}{\milli\second} in \name.
Additionally, \SI{25.4}{\percent} of relay pairs cut their RTT in half (or more) using \name.

\Cref{fig:rtt-dist}(b) also shows that \name's RTT values largely correspond the physical distance between the endpoints: relay pairs that are across an ocean necessarily experience a higher latency than those in the same region.

\subsection{\name Circuits}
\label{sec:perf}
We model the expected reduction in latency for end users by applying \name to Tor circuits.
Due to Tor's non-uniform relay selection probabilities, our pairwise latency dataset does not directly account for how \textit{probable} any of the observed RTTs are.
As such, we include an evaluation of \name on two million Tor circuits built according to Tor's default parameters. 
Because Tor averages 2M daily users, this roughly approximates the expected distribution of circuits over one day of use.

\begin{figure}[t]
    \centering
    \includegraphics[width=\linewidth]{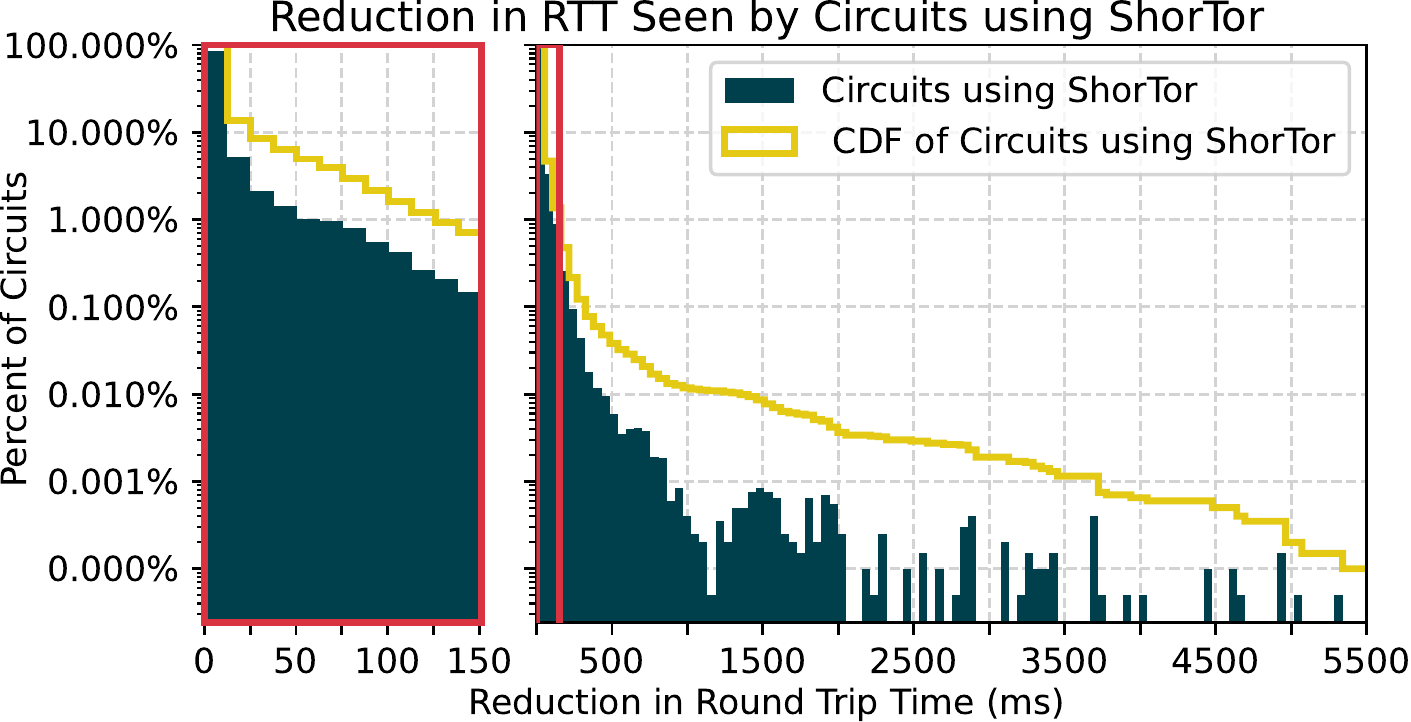}
    \caption{%
      Reduction in Round Trip Time experienced by 2M Tor circuits when routing with \name. We show an expanded view of the first \SI{150}{\milli\second} on the left.
      The CDF line is reversed: for instance, \SI{0.01}{\percent} of circuits see a speedup of at least \SI{1000}{\milli\second}.
    }
    \label{fig:rtt-dist-circ}
\end{figure}

\subsubsection{Circuit Selection}
\label{sec:circ-select}
Using the Tor path selection simulator (TorPS)~\cite{johnson2013users}, we choose two million circuits over \emph{all} \num{8524} relays we observed.
Because we collect latency measurements from only the \num{1000} largest Tor relays by consensus weight, many of these circuits have incomplete latency data.
We select circuits from the full set of relays, despite incomplete latency data, to ensure that our distribution of circuits closely resembles that of real Tor users.
We handle the gaps in our data by reporting on the \textit{reduction} in latency provided by \name rather than absolute RTTs.
All circuit legs with missing measurements are reported as a speedup of \SI{0}{\milli\second} (equivalent to baseline Tor).
The speedups we observe can thus be thought of as the \emph{minimum} that our set of simulated circuits would experience using \name.

\subsubsection{Latency of \name Circuits}
\label{sec:shortor-perf-comp}
In this section, we evaluate the performance of \name on our set of 2M circuits.
We only indicate a speedup for a circuit if: (1) it contains two adjacent relays that are present in both our measurement dataset and in the set of relays supporting \name and (2) some other relay that supports \name can provide a faster route between the circuit relays.
\SI{68.0}{\percent} of circuits have no available measurement data for either leg and are shown with the default speedup of \SI{0}{\milli\second}. 
Of the \SI{32.0}{\percent} of circuits with a latency measurement for at least one leg, \SI{83.7}{\percent} see a speedup with \name. 

As shown in \Cref{fig:rtt-dist-circ}, \SI{1}{\percent} of the 2M circuits see a latency improvement of \SI{122}{\milli\second} or greater and \SI{0.012}{\percent} of circuits saw a speedup of over a second.
For details on the relationship between RTTs and the page load times experienced by users, see \cref{sec:ux}

\subsubsection{Incremental Deployment}
\label{sec:eval-incremental}
As previously described in \Cref{sec:incremental}, \name is designed to function at relatively low levels of deployment. 
Our previous evaluation (\Cref{fig:rtt-dist-circ}) assumed that all \num{1000} of the relays we measured supported \name. 
In \Cref{fig:incremental-deployment}, we show that \name is also capable of reducing latency for Tor circuits even at substantially lower levels of deployment. 
As before, we only apply \name when all relays involved support the protocol and assume that all unmeasured pairs of relays have no speedup.
We find that circuits at the 99.9\textsuperscript{th} percentile see latency reductions of \SI{178}{\milli\second} even when only the top 500 relays support \name. 

\begin{figure}[t]
     \centering
     \includegraphics[width=\linewidth]{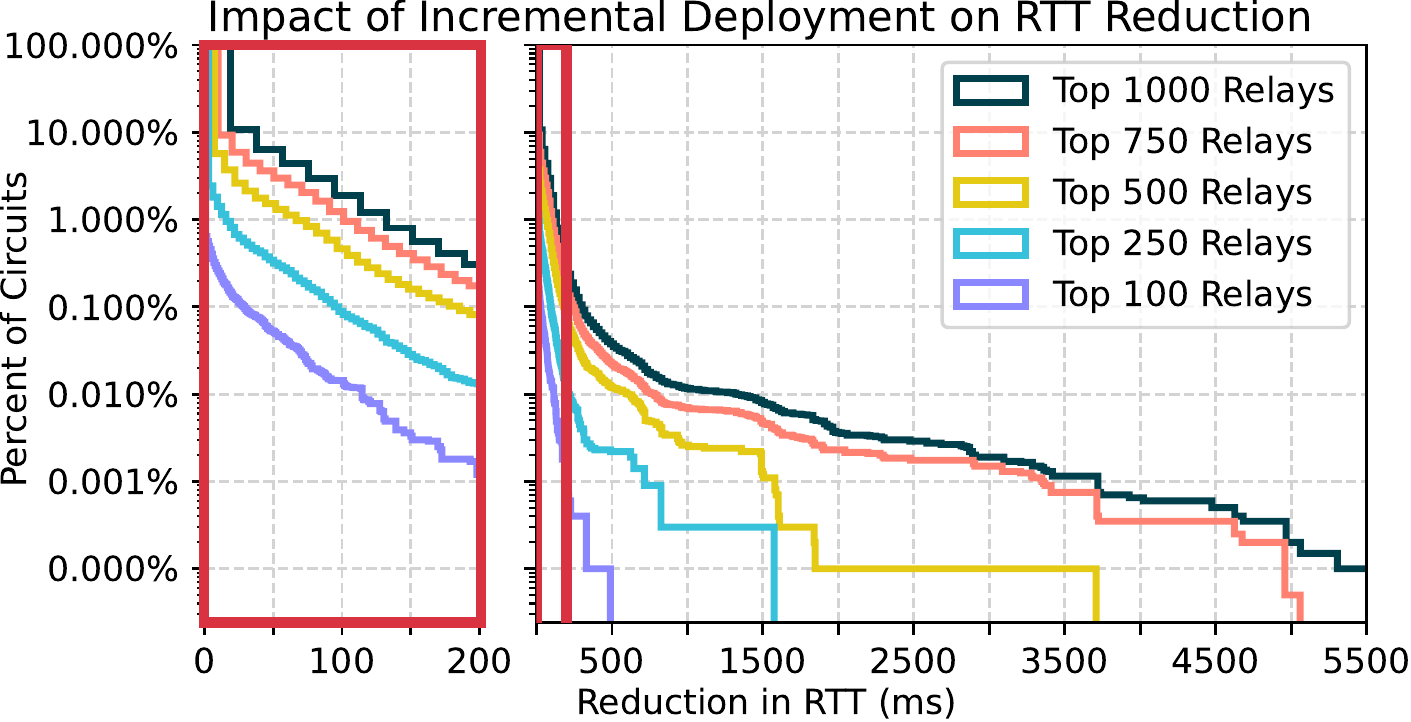}
     \caption{Reversed CDF for incremental deployment speedups.}
     \label{fig:incremental-deployment}
 \end{figure}

\subsection{Cost of \name}
\label{sec:shortor-cost}
\name's primary cost, aside from the one-time startup cost of modifying the Tor protocol, is in terms of bandwidth overhead from its longer paths.
In the steady state, Tor circuits will use extra bandwidth for each hop using a via relay: an overhead of 1/3 above the original traffic.
If we do this for \emph{every} hop with \emph{any} speedup, no matter how small, this uses \num{10.9}\% extra bandwidth over Tor right now. 
However, if we only use \name when it offers a speedup above a certain cutoff, this overhead quickly declines:

\begin{center}
\begin{tabular}{@{}lrrrrr@{}}
\toprule
\bf Cutoff & \bf \num{0} ms & \bf \num{10} ms & \bf \num{25} ms & \bf \num{50} ms & \bf \num{100} ms \\ 
Overhead & \num{10.9}\% & \num{6.6}\% & \num{3.8}\% & \num{2.3}\% & \num{0.8} \\
\bottomrule
\end{tabular}
\end{center}
Further, relays will only carry via traffic when they have excess capacity.
Tor reports consistently under 50\% bandwidth utilization across the network \cite{tor-metrics-network-size}.

A more minor source of bandwidth overhead is control information around routing tables and data races.
First, relays must keep their latency tables up to date following $\algname{Latencies.Update}()$ (\cref{proto:latency-measurements} in \cref{sec:pairwise-latency-alg}).
This requires each relay send its latency table to every other relay once per day.
Using 16 bits per RTT estimate, each relay must send about \SI{100}{\mega\byte} per day\footnote{Assuming 7000 concurrently active relays in the Tor network: 16 bits per RTT $\times$ 7000 relays per table $\times$ 7000 relays sending their table each day.} in total---about \num{0.05}\% of the minimum recommended relay bandwidth~\cite{tor-relay-requirements}.
Second, they must establish via connections using $\algname{Race.Run}()$ (\cref{proto:data-race} in \cref{sec:data-race}), which sends a small, configurable number of extra packets (e.g., 5), equivalent to about \SI{2.5}{\kilo\byte}\footnote{Tor cells are $\approx$ \SI{0.5}{\kilo\byte} each.} of extra data.
Data races have tunable frequency and will only occur if latency estimates indicate a speedup above the cutoff.
Assuming two million circuits per day with each circuit participating in a data race at both hops, data races will consume \SI{10}{\giga\byte} over the course of the day across the entire network. 
Currently, Tor advertises bandwidth of \SI{600}{\giga\byte}/s and consumes less than \SI{300}{\giga\byte}/s~\cite{tor-metrics}.

Thus, \name's bandwidth overhead is dominated by the longer path lengths and is parameterizable based on cutoffs for latency reduction as shown above.
This overhead may \emph{not} be distributed evenly among relays, but we note that participation in \name is fully optional (see discussion of incremental deployment in \Cref{sec:incremental}), so resource constrained relays may simply choose not to participate at any stage of \name or decide not to support the protocol entirely if overhead is a concern.
 
\subsection{Impact of \name on User Experience}
\label{sec:ux}
Perceived latency in the form of page load times (PLT) has a demonstrable impact on users in anonymity systems.
In qualitative user experience research, Tor users specifically cite latency as an issue keeping them from adopting Tor \cite{gallagher2017newme,gallagher2018peeling,gallagher2020measurement}.
\citet{kopsell2006low} finds a linear relationship between latency and number of users for an anonymous communication system.

In many cases, latency (not bandwidth) is the limiting factor for page loads: increases in RTTs cause linear increases to PLT, often with a \num{10}$\times$ multiplier.
\citet{netravali2018vesper} find that increasing RTT from \SI{25}{\milli\second} to \SI{50}{\milli\second} increases \num{95}th-percentile PLT across \num{350} popular sites from \SI{1.5}{\second} to \SI{3.4}{\second}, and increasing to \SI{100}{\milli\second} raises the PLT to \SI{6.1}{\second}.
Many factors contribute to this multiplier, including TCP congestion control, TLS handshakes, and complex web sites where an initial page fetch may spawn additional requests \cite{netravali2018vesper,rajiullah2019web,rajiullah2015towards}.

To bridge the gap between our RTT-based evaluation of \name and the more intuitive usability metric of PLTs, we simulate the impact of network delays on page load times over Tor, finding that small increases in delays lead to large increases in page load times.
First, we measure the time it takes to load the New York Times and Google homepages over ten Tor circuits, chosen by the Tor path selection algorithm, without modification or delay.
We then model changes in RTT, such as those from a link delay, by using the Linux \texttt{tc} utility to introduce an artificial delay for each packet sent over the same ten circuits.

To evaluate \name, we selected delays that correspond to speedups seen by \name circuits in \Cref{sec:perf} to obtain an estimate for the potential difference in page load times experienced by end users.
Of the circuits in our measurement dataset, \SI{5.04}{\percent} experience a speedup of at least \SI{50}{\milli\second}, \SI{1.66}{\percent} of at least \SI{100}{\milli\second}, and \SI{0.04}{\percent} of at least \SI{500}{\milli\second}.

We report the median change in PLT for fetching \texttt{google.com} and \texttt{nytimes.com} over these ten circuits when traffic is delayed by \SI{50}{\milli\second}, \SI{100}{\milli\second}, and \SI{500}{\milli\second}:

\medskip
\begin{tabular}{@{}llccc@{}}
              & &  \multicolumn{3}{c}{\bf Network Delay}   \\
              \cmidrule(lr){3-5}
{\bf Website} & & {\bf 50 ms}  & {\bf 100 ms}  & {\bf 500 ms} \\
\midrule
google.com & $\Delta$PLT:    & \cellcolor{red!5} 0.98 s & \cellcolor{red!10} 1.96 s & \cellcolor{red!20} 10.40 s\\
nytimes.com & $\Delta$PLT:    & \cellcolor{red!5} 1.66 s & \cellcolor{red!10} 2.34 s & \cellcolor{red!20} 15.80 s\\
\bottomrule
\end{tabular}
\medskip

We find that Tor follows the trend seen in prior work with even \SI{50}{\milli\second} changes in RTT increasing PLTs by approximately a second.
In the context of \name, \SI{1}{\percent}, or 20k, of the 2M circuits we evaluated saw a reduction in their RTT of at least \SI{120}{\milli\second} which corresponds to an expected two second drop in PLT. 
As Tor sees approximately two million daily users, each building at least one circuit, \name's impact on tail latencies is likely to improve the experience of tens of thousands of Tor users daily.

%% file: sections/security.tex
In this section, we analyze the security of \name. 
We consider how \name's use of via relays might impact an adversary's ability to deanonymize Tor traffic in practice. 
To do so, we examine the change in the adversary's view of the Tor network when using \name as compared to the baseline Tor protocol. 
While via relays never observe the sender or recipient of Tor traffic \emph{directly}, they are able to observe traffic streams and other relays on the circuit, which could \emph{indirectly} deanonymize the sender or recipient.
For this purpose, we use the AnoA~\cite{anoa,anoa-mators-thesis} framework for analyzing the anonymity guarantees of anonymous communication protocols to help us determine the potential anonymity impact of vias in \name. 

\subsection{AnoA Analysis}
\label{sec:anoa-and-mator}
\citet{mator,backes2014nothing} apply the AnoA framework to analyze the anonymity of the Tor network and the impact proposed protocol modifications might have on Tor's anonymity. 
AnoA uses ideas from differential privacy~\cite{dwork2008differential} to determine an adversary's advantage in a challenge-response game, which models the ability to distinguish between traffic streams.
In this game, the adversary statically corrupts a set of ``traffic observation points'' (i.e., Tor relays) and attempts to distinguish between two possible scenarios involving different senders and recipients for each traffic stream. 
The adversary's ability to distinguish between these scenarios models the overall anonymity of the Tor network.

\begin{definition}[Anonymity Notions~\cite{anoa}; simplified]
\label{def:anon-notions}
Let $\adv$ be a passive network adversary consisting of a set of corrupted relays and capable of observing a subset of network traffic through these relays. 
The anonymity notions are:
\begin{description}
    \item[Sender anonymity:] the probability that $\adv$ can distinguish between two potential senders of a given traffic stream. 
    
    \item[Recipient anonymity:]
    the probability that $\adv$ can distinguish between two potential recipients of a given traffic stream. 
    
    \item[Relationship anonymity:] 
    the probability that $\adv$ is capable of determining which sender is communicating with which recipient. 
    The anonymity game is defined for all pairs of senders (Alice and Bob) and recipients (Charlie and Diana); $\adv$ wins by successfully linking a traffic stream to the correct communicating pair. 
\end{description}
\end{definition}

Tor relays are not all created equal: stable, high bandwidth relays see a larger fraction of Tor's traffic, but are also more costly.
To accurately analyze an adversary's impact on anonymity, it is necessary to decide \emph{which} subset of relays are most beneficial to corrupt.
\citet{mator,backes2014nothing} develop \textsc{MaTor}~\cite{mator-tool} to model different adversarial corruption strategies.
Following \citet{mator,backes2014nothing}, we consider four adversarial corruption strategies: $k$-collusion, bandwidth, monetary, and geographic.
With the $k$-collusion strategy, the adversary corrupts $k$ relays that provide it with the most advantageous view of Tor's network.
With the other three strategies, the adversary has a similarly fixed  ``budget'' (e.g., cumulative bandwidth) which constrains the optimal set of relays to corrupt. 
The \textsc{MaTor} tool~\cite{mator-tool} optimizes each adversarial strategy based on the allocated budget and anonymity notion to empirically compute the worst-case anonymity bound under AnoA.

\subsection{Differential Advantage}
\label{sec:differential-advantage}
We define how we measure the theoretical impact of \name on anonymity in Tor in terms of the difference in an adversary's advantage in \name vs. baseline Tor.

\paragraph{Notation}
We write $\nodeuni$ for the set of all Tor relays and $\circuituni$ for the set of all Tor circuits (consisting of three independent relays).
We denote the path selection and via selection algorithms by $\mathsf{PS}$, and $\mathsf{VS}$, respectively.
We note that $\mathsf{VS}$ is unique to \name and is separate from the path selection algorithm $\mathsf{PS}$ used in Tor.
We use $\bot$ for a ``null'' element. 

\begin{definition}[Via Relay]
Let $\circuit \in \nodeuni\times\nodeuni\times\nodeuni$ be a Tor circuit consisting of three circuit relays. 
A \emph{via relay} $\vianode \in \nodeuni$ is a Tor relay routing packets between a pair of consecutive circuit relays in $\circuit$.  
\end{definition}

\begin{remark}
\label{remark:via-and-wire-equiv}
A via relay is semantically equivalent to a \emph{wire} connecting two consecutive Tor relays in the circuit. 	
Via relays only forward traffic and are \emph{not} involved in circuit establishment or any of Tor's onion-encryption operations.
\end{remark}

We first define \emph{adversary observations} on the network.
We then use this to define the \emph{differential advantage}---the impact that \name introduces relative to baseline Tor.

\begin{definition}[Adversary Observations]
\label{def:obs-points}
Let $\viauni \subseteq \nodeuni$ be the set of candidate via relays and let $\circuituni$ be the set of all three-relay circuits. 
Fix a set $\nodeuni^* \subseteq \nodeuni$  of adversary-corrupted relays.
We define the function:
$\obs_{\nodeuni^*}: \circuituni \times \viauni \cup  \set{\bot}\times \viauni \cup \set{\bot} \to \obset$,
which takes as input a circuit and a pair of via relays (possibly $\bot$), and outputs the set of observation points (adversary-corrupted circuit and via relays). 

\end{definition}
\Cref{def:obs-points} captures the ``view'' of the adversary for a given circuit. 
For example, an adversary corrupting the middle relay on a single circuit sees the guard and exit relays on the circuit, but not the sender or recipient.

\begin{definition}[Differential Advantage]
\label{def:diff-adv}
Let $\viauni \subseteq \nodeuni$ be the set of candidate via relays.
Let $\mathsf{PS}$ be a randomized path selection algorithm and $\mathsf{VS}$ be a via relay selection algorithm for a circuit.  
Fix $\nodeuni^* \subseteq \nodeuni$ to be the set of adversary-corrupted relays and let $\circuituni$ be a set of all three-relay circuits output by $\mathsf{PS}$. 
For a circuit $\circuit \in \circuituni$ and $(\vianode_1, \vianode_2) \in \mathsf{supp}_{\circuit \in \circuituni}|\mathsf{VS}(\circuit)|$, the adversary is said to have a \emph{differential advantage} when for 
\begin{align*}
    \obset_{\mathrm{tor}} \gets \obs_{\nodeuni^*}(\circuit, \bot, \bot) \text{ and } \obset_{\mathrm{via}} \gets \obs_{\nodeuni^*}(\circuit, \vianode_1, \vianode_2),
\end{align*}
the set $\obset_{\mathrm{tor}} \subset \obset_{\mathrm{via}}$,
where $\obs$ is as defined in \Cref{def:obs-points}.  
\end{definition}

\noindent
In words, an adversary has a differential advantage in \name when \emph{new} observations are gained as a result of introducing via relays.
We now examine scenarios in which the adversary \emph{does} gain differential advantage by corrupting via relays. 
We formalize these scenarios in \Cref{thm:advantagious-scenario}.

\begin{lemma}
\label{thm:advantagious-scenario}
Let $\viauni \subseteq \nodeuni$ be the set of candidate via relays.
Fix $\nodeuni^* \subseteq \nodeuni$, the set of adversary-corrupted relays. 
For all sets of observations $\obset_{\mathrm{tor}}$ and $\obset_{\mathrm{via}}$ for a circuit $\circuit \in \circuituni$, as in~\cref{def:diff-adv},  $\obset_{\mathrm{tor}} \subset \obset_{\mathrm{via}}$ if and only if there exists at least one via relay in $\nodeuni^*$ between two consecutive non-corrupted relays in $\circuit$. 
\end{lemma}
\begin{proof}
	Let $\node_{a},\node_{b} \in \circuit$ be any two consecutive circuit relays in $\circuit$ (either $\{\mathrm{guard},\mathrm{middle}\}$ or $\{\mathrm{middle},\mathrm{exit}\}$) with via relay $\vianode$ connecting $\node_{a}$ and $\node_{b}$. 
	Corrupting either $\node_a$ or $\node_b$ provides the adversary with a view of the wire, which is equivalent to the view obtained from corrupting the via (see \Cref{remark:via-and-wire-equiv}). 
	For any circuit $\circuit \in \circuituni$, the set of observation points gained from corrupting $\vianode$ is a strict subset of the set of observation points gained from corrupting either $\node_{a}$ or $\node_{b}$ individually. 
	Therefore, we have that the adversary only obtains an additional observation ($\obset_{\mathrm{tor}} \subset \obset_{\mathrm{via}}$) if $\node_{a}$ and $\node_{b}$ are \emph{not} corrupted while the via relay $\vianode$ \emph{is} corrupted.
\end{proof}

\begin{claim}
\label{cor:middle-node-equiv}
An adversary-corrupted via relay observes strictly less than an adversary-corrupted circuit middle relay in Tor. 
\end{claim}
\begin{proof}
By \cref{thm:advantagious-scenario}, we have that the adversarial advantage from corrupting a via relay is strictly less than corrupting any middle relay.
Via relays are positioned either between the guard and middle relays or middle and exit relays. 
As such, corrupting a middle relay in a circuit tightly upper bounds the observation points gained from corrupting both vias. 
\end{proof}

In \cref{prop:shortor-no-advantage}, we argue that \name does not advantage the adversary in any of the anonymity notions of \cref{def:anon-notions} (we empirically confirm this result in \cref{sec:quantifying-advantage}).

\begin{claim}
\label{prop:shortor-no-advantage}
\name applied to the baseline Tor network with path selection algorithm $\mathsf{PS}: \nodeuni \to \circuituni$ which outputs circuits independently of the sender and recipient (as is currently the case in Tor~\cite{tor-spec,tor-paper}), does not impact the anonymity notions of \Cref{def:anon-notions} of the AnoA framework.
\end{claim}
\begin{proof}
Under the AnoA framework, corrupting the middle relay does not change the adversary's ability to deanonymize either sender, recipient, or relationship anonymity when the circuit is constructed \emph{independently} of the sender and recipient (see analysis of \citet{mator}).
This is because each middle relay is equally probable in all communication scenarios, giving the adversary no advantage~\cite{mator}.
As a consequence, by leveraging \cref{cor:middle-node-equiv}, corrupting one or both via relays when \name is applied to Tor does not advantage the adversary in the AnoA anonymity game of \cref{def:anon-notions}.  
\end{proof}

\Cref{prop:shortor-no-advantage} shows that \name does not impact anonymity of Tor.
However, when the middle relay is \emph{not} chosen independently of the sender or recipient (for example, when using location-aware path selection proposals; see \cref{sec:relatedwork}), then \name can exacerbate the negative impact on anonymity.

We quantify this advantage using \textsc{MaTor} by applying \name to LASTor~\cite{akhoondi2012lastor}, a location-biased path-selection proposal.
We emphasize that LASTor is \emph{not} integrated in Tor and has known security flaws~\cite{Wan2019guard}---we include it as an illustrative example of a location-aware path selection scheme.

\begin{figure}[t]
\centering
\subfloat{{\includegraphics[width=0.48\linewidth]{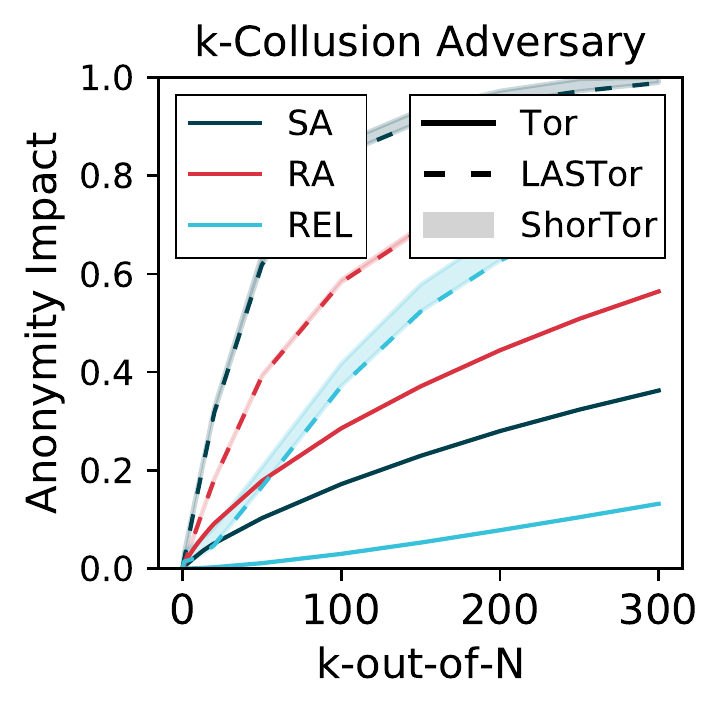}}}
\subfloat{{\includegraphics[width=0.48\linewidth]{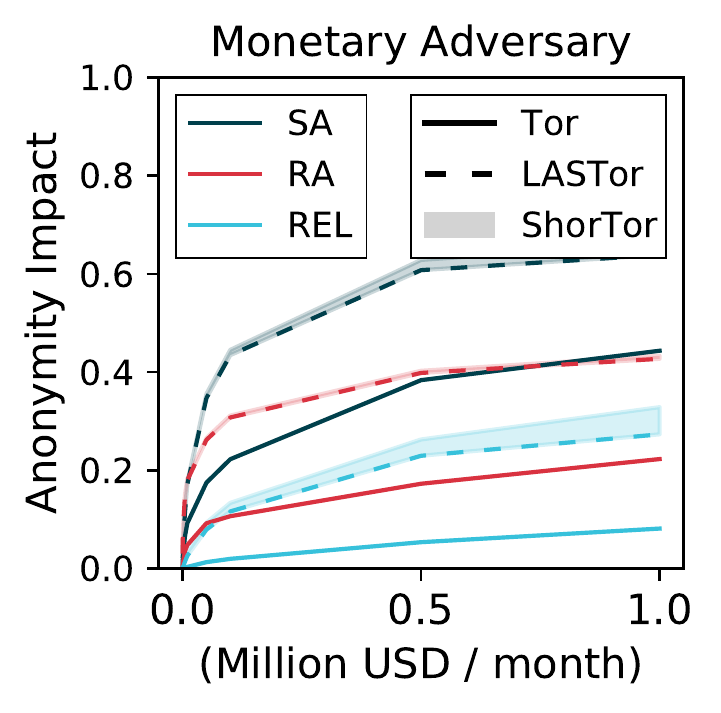}}}\\
\subfloat{{\includegraphics[width=0.48\linewidth]{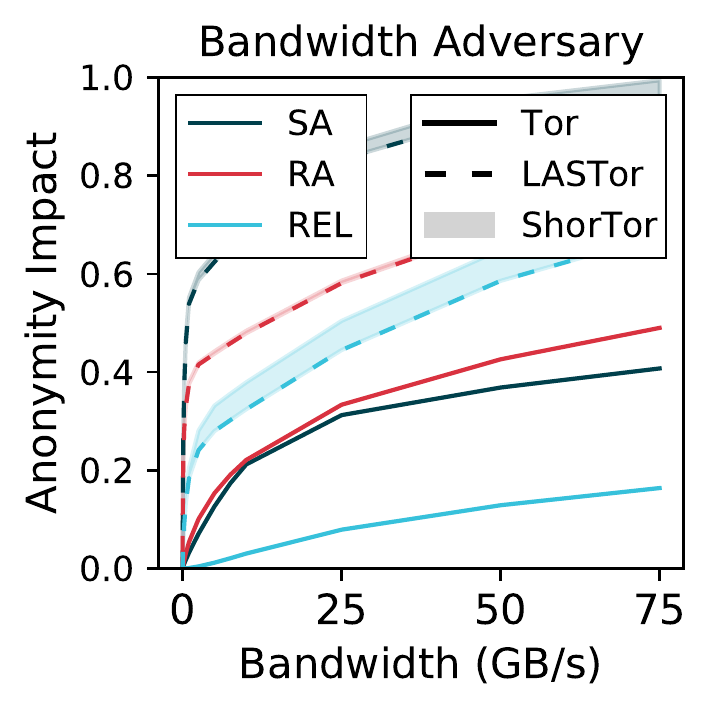}}}
\subfloat{{\includegraphics[width=0.48\linewidth]{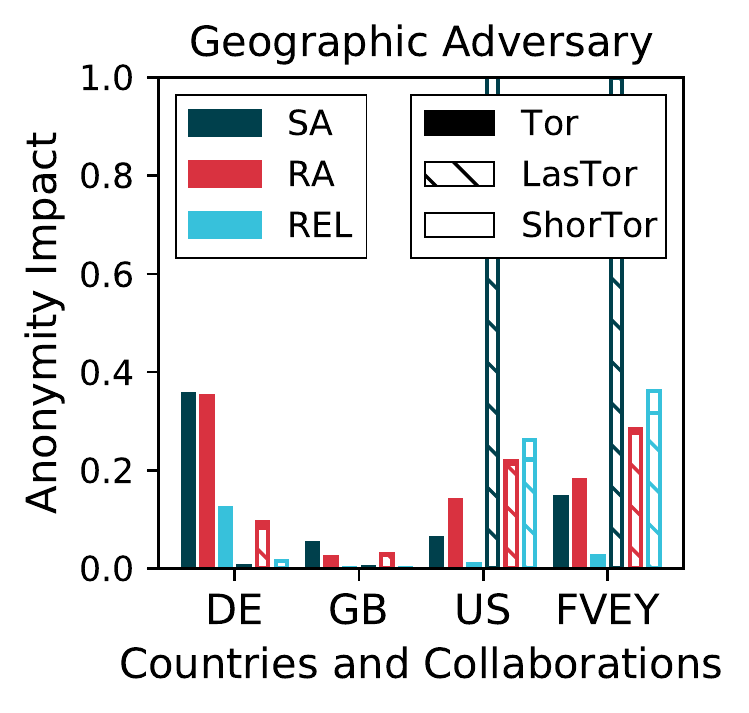}}}
\caption{
Anonymity impact of \name, compared with baseline Tor (client-independent) and LASTor (dependent on client location) path selection.
Each plot shows sender (SA), recipient (RA), and relationship (REL) anonymity (\Cref{def:anon-notions}) for a different adversarial relay corruption strategy. 
Shaded region represents the difference in the \textsc{MaTor}-computed anonymity bounds with and without \name. 
``FVEY'' refers to ``Five Eyes'' intelligence alliance member countries.
Under ANoA, \name affects all anonymity notions for LASTor, though not baseline Tor. Extended plots provided in \cref{sec:extended-mator-plots}.
}
\label{fig:mator}
\end{figure}

\subsection{Quantifying Anonymity of \name}
\label{sec:quantifying-advantage}
We now turn to empirically computing the \emph{worst-case} anonymity impact under the AnoA framework  when \name is applied to Tor and proposed modifications thereof. 
We modify \textsc{MaTor} (our code is open-source~\cite{source-code-for-shortor}; written in C++ and Python) to incorporate the use of via relays as described in \cref{sec:shortor-proto}. 
We report our quantitative results in \cref{fig:mator}. 

\paragraph{\name applied to Tor}
We confirm the results of \cref{prop:shortor-no-advantage} on the adversarial impact of \name used with baseline Tor:
the worst-case anonymity bounds computed by \textsc{MaTor} for baseline Tor and \name are equal, as relays are selected independently of both the sender and recipient.

\paragraph{\name Applied to LASTor}
We examine the impact of \name when combined with \emph{biased} path selection algorithms (e.g., path selection that takes client location into account). 
We use the LASTor~\cite{akhoondi2012lastor} proposal for this purpose. 
We find that \name applied to LASTor decreases anonymity under all three anonymity notions of \Cref{def:anon-notions}, as via relays offer additional observations points for the skewed distribution of guards and exits used by LASTor. 

\begin{figure}[t]
\centering
\includegraphics[width=\linewidth]{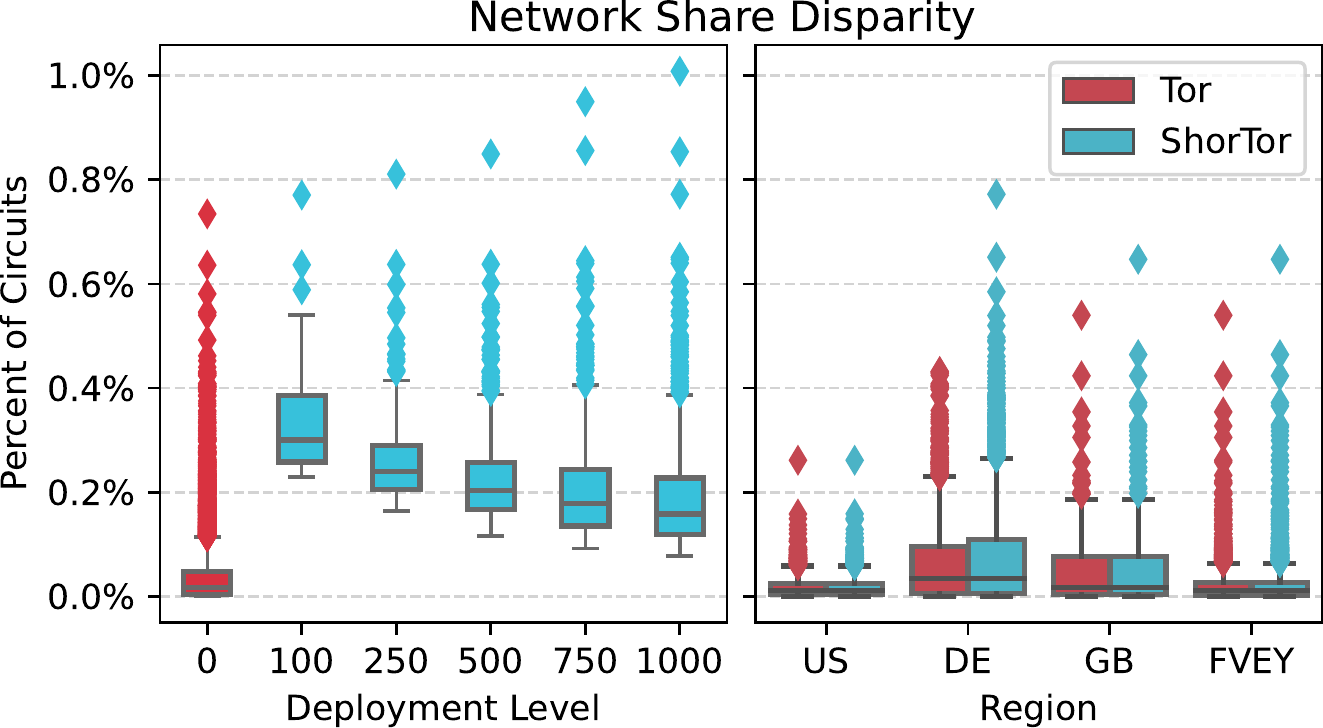}
\caption{Fraction of circuits seen by relays in Tor vs. \name. 
Left side considers incremental deployment (\cref{sec:eval-incremental}).
As more relays begin to support \name (x-axis), the average network share goes down, though some outliers see more traffic.
Right side looks at regional network share, or the fraction of circuits seen by relays in different regions.
}
\label{fig:network-share}
\end{figure}

\subsection{Network Share and Traffic Analysis Attacks}
\label{sec:network-share-analysis}
A limitation of AnoA is that it does not take into account traffic analysis attacks that can be conducted by a single point of observation~\cite{anoa-mators-thesis}, such as a via relay.
Most notably, this includes website fingerprinting attacks~\cite{jansen2018inside,fingerprinting-usenix,Bhat_2019,Rimmer_2018,sirinam2018deep}.

To analyze the adversarial advantage in orchestrating such attacks in \name, we consider the relative \emph{network share} disparity between baseline Tor and \name. 
By this we mean the relative fraction of circuits seen by a relay when acting as part of a circuit vs. as a via.
Regional network share is a concern for Tor users primarily due to varying policies on surveillance in different jurisdictions~\cite{johnson2017avoiding}.
Separately, determining network share of individual relays at different levels of deployment allows us to assess the potential security impact of incrementally deploying \name.

In \cref{fig:network-share}, we plot the network share using the same circuits as in \cref{sec:shortor-perf-comp}.
We vary the deployment level of \name to measure the expected change in network share as a function of relays supporting \name (\cref{sec:eval-incremental}).
We find that \name increases the median network share (as traffic traverses more nodes when using \name).
However, median network share decreases with larger deployments.
The worst-case network share increases from about \SI{0.4}{\percent} to about \SI{0.8}{\percent} for relays located in Germany (which is by far the country with the most Tor relays).
It is important to note that the overall network shares remain low, indicating a small disparity in expected network traffic observed.

Our analysis does not take into account \emph{adversarial} placement of relays with fast network connections to boost their via selection probability.
However, this is not unique to \name: relays in Tor already have a high disparity in their network share based on \emph{bandwidth} influencing their circuit selection probability. 
As such, we believe that the impact of \name on traffic analysis attacks is modest and in-line with Tor's existing assumptions about adversarial placement of relays.

%% file: sections/relatedwork.tex
In this section, we outline past works that focus on reducing latency in Tor though optimized routing decisions.
We note that all works here operate at the \textit{circuit} layer and are proposed modifications to Tor's circuit selection protocol.
There is additionally a large body of work that alters path selection in Tor for purposes of \emph{security}~\cite{yang2015enhancing,nithyan2015measuring,sun2017counter,edman2009awareness,rochet2016waterfiling,johnson2017avoiding,barton2016denasa,hanley2019dpselect,backes2014nothing}.
While important, these works are orthogonal to \name and often result in substantially degraded performance~\cite{rochet2020claps, mitseva2020security} without clear security advantages over Tor's current protocol~\cite{Wan2019guard,wacek2013empirical,geddes2013how}.

\paragraph{Traffic Splitting}
Rather than selecting a \textit{single} faster circuit, \citet{karaoglu2012multi} and \citet{alsabah2013path} split traffic across multiple circuits.
This distributes the load of the circuit across a larger number of relays, improving latency by reducing congestion on relays in the circuit.
Conflux~\cite{alsabah2013path} can achieve an average reduction in time-to-first-byte of \SI{23}{\percent} over baseline Tor.
Splitting traffic across multiple circuits solves an orthogonal problem to that addressed in \name and combining both protocols could be an interesting future direction.

\paragraph{Location-Aware Path Selection}
Alternative path selection proposals that reduce latency on Tor share a common theme: they all, either directly or indirectly, account for the location of the client or destination server when choosing a circuit~\cite{imani2019modified,wang2012congestion,annessi2016navigator,rochet2020claps,akhoondi2012lastor,sherr2009scalable,barton2018towards}.
This makes intuitive sense, as fast paths are likely to also be geographically short and, in particular, are unlikely to contain multiple ocean crossings.

\citet{imani2019modified} propose to improve performance of Tor circuits by having clients build multiple circuits, then preferentially select from those according to a series of strategies focusing on circuit length, RTT, and congestion.
In addition to latency, \citet{sherr2009scalable} include measurements of jitter and packet loss when selecting relays. 

\citet{wang2012congestion} opportunistically selects relays with low latency to construct circuits that avoid congested relays. 
NavigaTor~\cite{annessi2016navigator} applies a similar strategy to \citet{wang2012congestion} and demonstrates improved performance by using latency (specifically round-trip time) to discard slow Tor circuits.
PredicTor~\cite{barton2018towards} avoids the overhead of constructing then discarding multiple circuits by using a random forest classifier trained on Tor performance data to \textit{predict} the performance of a circuit prior to building it. 
CLAPS (CLient Aware Path Selection)~\cite{rochet2020claps} solves a weight-optimization problem to ensure a strict bound on anonymity degradation when selecting location-biased circuits. 
There additionally exists a large body of work on the security implications of location awareness in path selection as this can make Tor clients more identifiable by creating a correlation between their geographic location and the relays chosen for their circuit~\cite{rochet2020claps,mitseva2020security,mator,anoa-mators-thesis,anoa,geddes2013how,wacek2013empirical,Wan2019guard}.
As shown in \cref{sec:secanalysis}, \name does not share this problem and is able to reduce latency for Tor clients without the security pitfalls of location-aware circuit selection techniques.

%% file: sections/discussion.tex
\paragraph{Other sources of delay in Tor}
The type of overhead that \name addresses is not the largest source of delay in Tor. 
Limited congestion control~\cite{alsabah2011defenestrator,fiedler2020predictor}, under-optimized multiplexing of circuits on TCP connections~\cite{baysoni2019quick,baysoni2021quictor}, and high queuing delays~\cite{jansen2018kist,jansen2014never,jansen2017tor} are likely larger contributors to latency in Tor than suboptimal BGP routes.
Despite this, we believe \name to be of interest.
The source of delays addressed by \name and the techniques applied are both completely independent of other delays in Tor.
Because of this, the decrease in latency provided by \name will trivially \textit{stack} with any future improvements to congestion control, circuit multiplexing, or queueing.
As such, we believe \name to be a valuable contribution to improving the latency of Tor connections.

\paragraph{Compatibility with security-focused path selection}
\name is also fully compatible with any modifications to Tor's path selection algorithm.
 Prior work has shown that existing proposals, overviewed briefly in \Cref{sec:relatedwork}, suffer from poor load balancing and non-uniform client behavior, hurting performance and client anonymity~\cite{rochet2020claps,mitseva2020security,mator,anoa-mators-thesis,anoa,geddes2013how,wacek2013empirical,Wan2019guard}. 
However, this does not preclude some \emph{future} path selection proposal from improving upon Tor's current algorithm.
In this case, \name is again agnostic to the choice of the path selection algorithm and would require no modification to continue improving latency on top of the new algorithm.

\paragraph{Generality} While we apply multi-hop overlay routing to Tor specifically, we note that it is a general purpose technique.
Evaluating its effectiveness for other relatively small scale, distributed communication networks is an interesting direction for future work.
However, as shown by prior work~\cite{schlinker2019peering} and confirmed here for Tor, accurate evaluation of multi-hop overlay routing cannot be done with general purpose latency data and requires measurements from the specific network involved.

%% file: sections/conclusion.tex
We presented \name, an incrementally-deployable protocol for improving the latency of Tor's connections. 
We evaluated the performance and security of \name, demonstrating that it provides substantial improvements to tail latencies on Tor circuits, with minimal impact to security.
As part of our evaluation we collected a dataset of pairwise latencies between the thousand most popular Tor relays. 
This dataset allowed us to determine the reduction in latency \name provides to Tor circuits \textit{directly} without relying on simulation or approximated data.
Finally, while we proposed and evaluated \name specifically for Tor, the protocol is general and has foreseeable applications to other distributed communication networks.

%% file: sections/appendix.tex
\section{Extended \textsc{MaTor} Plots}
\label{sec:extended-mator-plots}
In this section we provide additional data from our \textsc{MaTor} analysis described in \Cref{sec:quantifying-advantage}. 
We refer to \Cref{sec:secanalysis} for details of the analysis and results. 
\begin{figure}[H]
\centering
\subfloat{{\includegraphics[width=\linewidth]{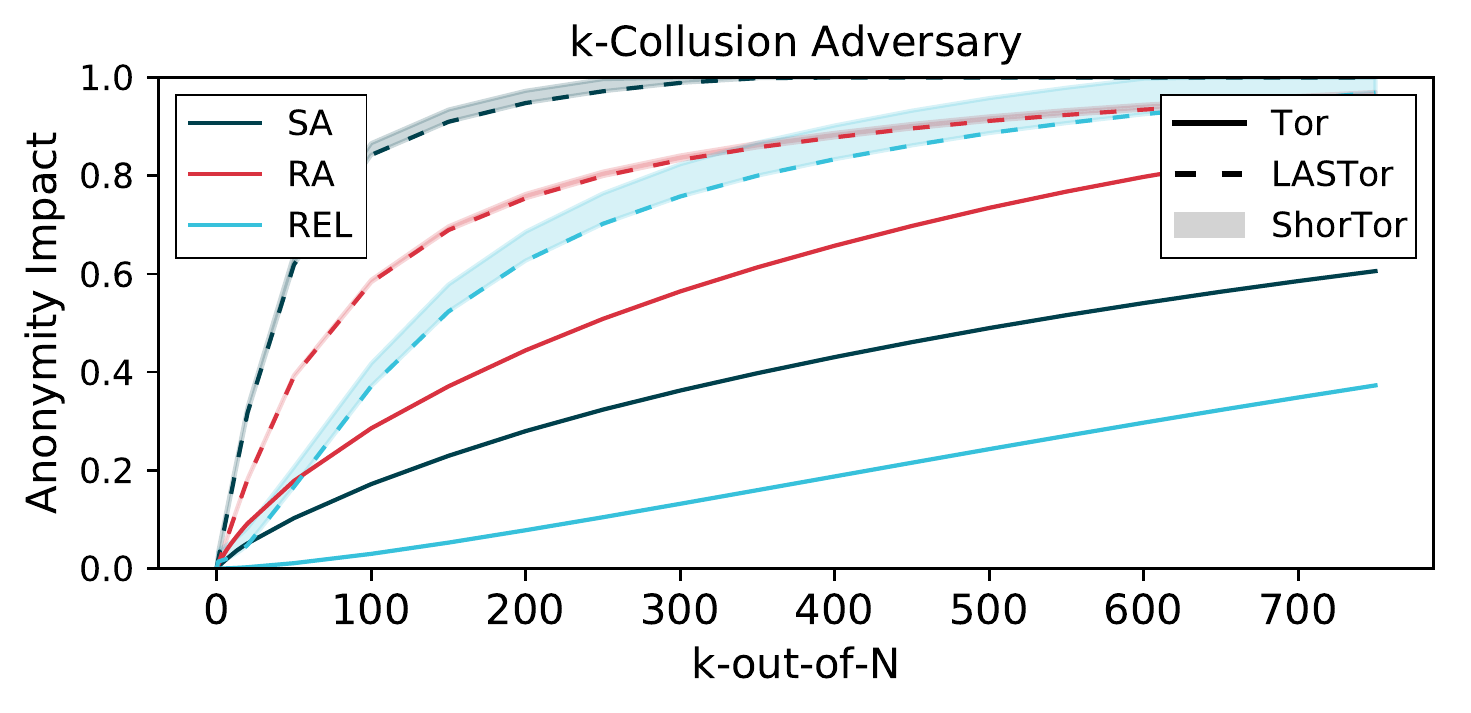}}}\\
\subfloat{{\includegraphics[width=\linewidth]{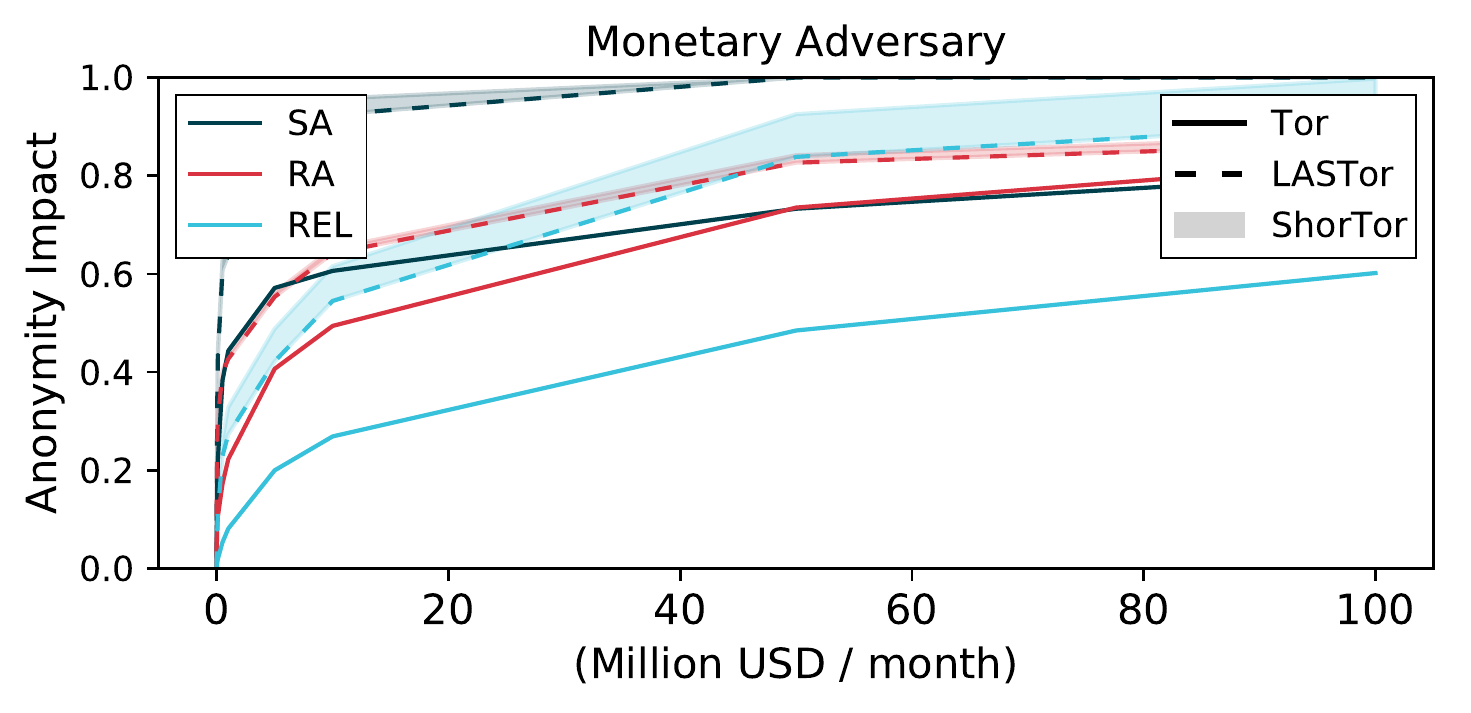}}} \\
\subfloat{{\includegraphics[width=\linewidth]{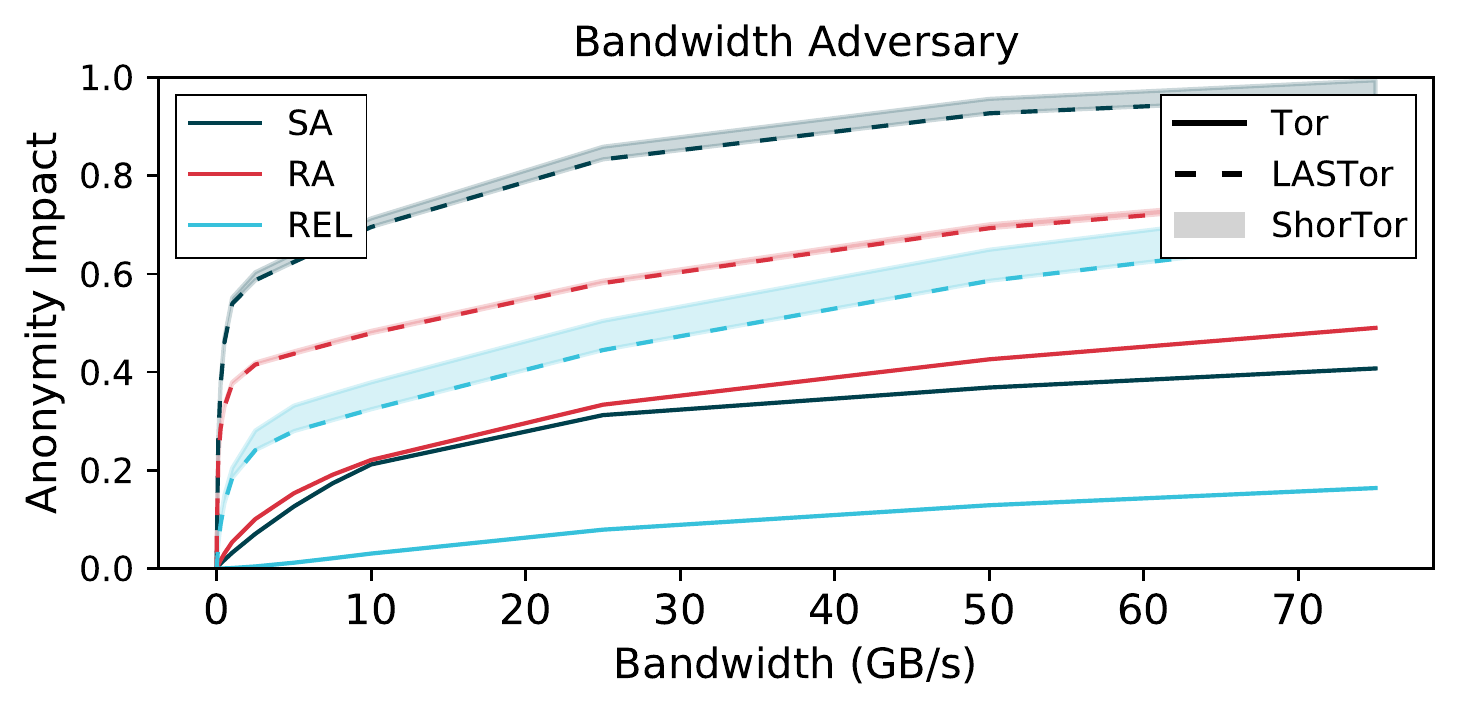}}} \\
\subfloat{{\includegraphics[width=\linewidth]{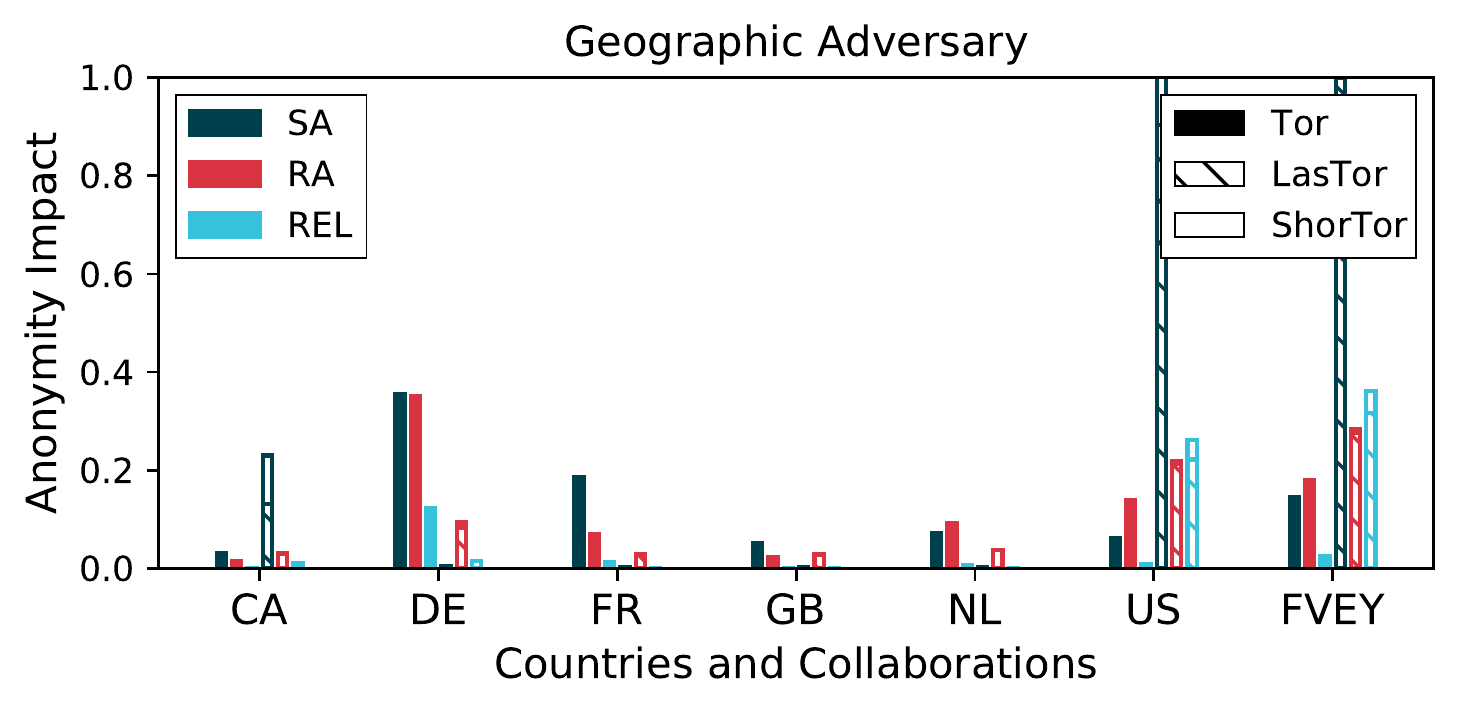}}}
\caption{
{\bf Extended experiments} for the anonymity impact of \name.
Anonymity impact of \name, compared with baseline Tor (client-independent) and LASTor (dependent on client location) path selection.
Each plot shows sender (SA), recipient (RA), and relationship (REL) anonymity (\Cref{def:anon-notions}) for a different adversarial relay corruption strategy. 
Shaded region represents the difference in the \textsc{MaTor}-computed anonymity bounds with and without \name. 
``FVEY'' refers to ``Five Eyes'' intelligence alliance member countries.
Under ANoA, \name affects all anonymity notions for LASTor, though not baseline Tor.
}
\label{fig:mator-extended}
\end{figure}

\section{Extended Integration with Tor}
\label{appendix:tor-modifications}
\Cref{sec:tor-integration} covers the high-level design decisions we made for \name to support integration with Tor.
Here, we provide specific details on the necessary modifications and extensions which are required to integrate \name with Tor. 
\name has three primary components not present in baseline Tor: (1) on-demand data races, (2) periodic sharing of latency information, and (3) a separate data path for via traffic.
We design these components to be minimal extensions to Tor's existing protocol.

\subsection{Data Races} 
In \name, data races represent the majority of required modifications.
To successfully conduct a data race for traffic on a circuit, relays must: (1) recognize data race traffic as separate from steady state traffic, (2) decide when a race should be conducted while observing appropriate backoff parameters, (3) interpret the new latency tables to discover potential vias, and (4) update routing information to include vias. 

We now describe how these four requirements can be implemented into Tor relay logic.

\begin{enumerate}[label=(\arabic*)]
    \item 
        can be achieved using the existing \texttt{CMD} field in Tor's cell header (see \Cref{fig:via-cell}). 
        By introducing a new \texttt{CMD} value to indicate that traffic is part of a data race, both vias and circuit relays can immediately recognize (and potentially drop) race traffic with minimal processing. 
        Importantly, using the \texttt{CMD} field allows \name to conduct data races \emph{without} altering the content of the cell. 
        Because of this, data race cells are simply normal cells from a client's traffic stream and clients will see no interruption or delay while the race is run.
    \item 
    is described in \Cref{sec:shortor-proto} and \Cref{sec:tor-integration}. 
    The only additional detail is that this process will need to tie into the queuing architecture shown in \Cref{fig:tor-queuing} such that, when a data race is to be conducted for a circuit, the next cell out of its queue will be duplicated, have its header modified according to \Cref{fig:via-cell}, and have the copy rerouted to each of the vias in the race prior to reaching the output buffer.
    This is due to the fact that, as the cells are now going to different destinations, they will also be on different TCP/TLS connections and must be sent to the appropriate output buffer.
    \item 
    is a simple matter of reading the latest latency table and selecting the vias with the lowest recorded latencies to participate in the race. 
    
    \item 
    is handled by a new field for the \texttt{ViaID} in Tor's routing table and the new header information in Tor cells containing the \texttt{IDs} of the two adjacent circuit relays (see \Cref{fig:via-cell}).
    This information allows the via relay and both adjacent circuit relays to recognize which circuit a traffic stream is from (and, consequently, where it should be routed) even when the traffic stream arrives on a different TCP/TLS connection than is usual for the circuit (i.e., arrives through a via, not from the previous circuit relay).
\end{enumerate}

\subsection{Latency Tables}
Maintaining and disseminating up to date latency information is a simple process detailed in \Cref{proto:latency-measurements} presented in Section \ref{sec:shortor-proto}.
Unlike data races, it can be run entirely independently from the main Tor protocol and does not involve any circuit traffic.
As such, it requires no actual integration with the Tor protocol, but is simply an entirely separate functionality run by Tor relays.
The table of latency information produced as part of this process is accessed to inform the selection of candidate vias for a data race, but is not otherwise involved and, in particular, is not used during the steady-state of routing traffic through vias.

\subsection{Data Path for Via Traffic}
The data path for via traffic is almost identical to that of regular, non-via, Tor traffic and follows the queuing architecture of baseline Tor as shown in \Cref{fig:tor-queuing}.
The two differences are that traffic is not onion encrypted/decrypted while being forwarded by a via and that via traffic is considered lower priority than circuit traffic. 
Via traffic, like data races, can utilize the \texttt{CMD} field in the cell header to identify itself as soon as the TLS layer has been decrpyted. 
This lets the via relay know that it is not expected to operate on the onion encryption layer (and does not have the keys necessary to do so) and that it should send the cell directly to the appropriate queue instead.
Finally, while Tor already prioritizes browsing traffic over bulk download traffic, \name requires a new priority level for via traffic that is below that of circuit traffic.
This priority level is applied when cells are dequeued to be scheduled to the output buffers.
Deprioritizing via traffic is necessary to ensure that it never outcompetes circuit traffic sharing the same relay (see paragraph in \Cref{sec:tor-integration} on load balancing and fairness).